\documentclass[11pt]{article}
\usepackage{amsfonts}
\usepackage{amsthm}%
\usepackage{mathrsfs}%
\usepackage{mathrsfs}%
\usepackage[fleqn]{amsmath}    
\usepackage{amssymb}
\usepackage{dsfont}
\usepackage{hyperref}
    \hypersetup{
    colorlinks=true,%
    citecolor=black,%
    filecolor=black,%
    linkcolor=black,%
    urlcolor=black,%
    }
\usepackage[top=1.75cm,bottom=1.75cm,left=2cm,right=2cm]{geometry}
\usepackage{setspace}
\usepackage{enumitem}
\usepackage{float}
    \restylefloat{table}
\usepackage{graphicx}
\usepackage{epstopdf}
\usepackage{caption}
\usepackage{subcaption}

\renewenvironment{abstract}
{\small
\begin{center}
\bfseries\abstractname\vspace{-.5em}\vspace{0pt}
\end{center}
\list{}{%
\setlength{\leftmargin}{0.75cm}%
\setlength{\rightmargin}{\leftmargin}%
}%
\item\relax}
{\endlist}

\newtheorem{thm}{Theorem}[section]
\newtheorem{lem}{Lemma}[section]
\newtheorem{prop}{Proposition}[section]
\newtheorem{rem}{Remark}[section]

\newtheorem{defi}{Definition}[section]

\numberwithin{equation}{section}

\title{Rapid stabilization of the heat equation\\ with localized disturbance}

\author{Christian Calle \footnotemark[1]\hspace{0.275cm} Patricio Guzm\'an \footnotemark[1]\hspace{0.275cm} and Hugo Parada \footnotemark[2]\hspace{0.275cm}}

\begin{document}
\setlength{\abovedisplayskip}{0pt}
\renewcommand\qedsymbol{$\blacksquare$}   

\footnotetext[1]{Departamento de Matem\'atica, Universidad T\'ecnica Federico Santa Mar\'ia, Valpara\'iso, Chile.\\
E-mail: {\tt patricio.guzmanm@usm.cl, ccalle@usm.cl  (corresponding author)}}
\footnotetext[2]{Universit\'e de Lorraine, CNRS, Inria, IECL, F-54000 Nancy, France.\\
E-mail: {\tt hugo.parada@inria.fr}}

\date{}

\maketitle

\begin{abstract}
This paper studies the rapid stabilization of a multidimensional heat equation in the presence of an unknown spatially localized disturbance. A novel multivalued feedback control strategy is proposed, which synthesizes the frequency Lyapunov method (introduced by Xiang \cite{X2024}) with the sign multivalued operator. This methodology connects Lyapunov-based stability analysis with spectral inequalities, while the inclusion of the sign operator ensures robustness against the disturbance. The closed-loop system is governed by a differential inclusion, for which well-posedness is proved via the theory of maximal monotone operators. This approach not only guarantees exponential stabilization but also circumvents the need for explicit disturbance modeling or estimation.
\end{abstract}

\vspace{0.3cm}


\vspace{0.5cm}




\section{Introduction}

\hspace{0.75cm}Partial differential equations (PDEs) play a central role in the mathematical modeling of a wide range of physical phenomena, such as heat distribution in solids and fluids, wave propagation, and the lateral deflection of strings and beams. Once a model is formulated, a fundamental objective in control theory is to design feedback laws that stabilize the system’s state, either toward an equilibrium or a desired trajectory. The traditional stabilization analysis often proceeds under idealized assumptions, with the absence of external disturbances. However, in practical applications, systems are subject to disturbances arising from unmodeled fast dynamics, parameter uncertainties, or fluctuating environmental loads, for instance. Those disturbances can affect the stability of the system. Similarly, control designs that require actuation across the entire spatial domain are often physically unrealizable. Instead, controls might be applied through a specific subregion of the domain or on a section of its boundary. Consequently, a central and challenging objective in PDE control is to develop stabilization strategies that are both robust to disturbances and spatially localized in their actuation.


Stabilization problems for one-dimensional PDEs, such as controlling a string or a rod, are well-studied. There are numerous results for various boundary conditions and actuator configurations, where it is common to employ powerful methods like backstepping, which has proven effective for many one-dimensional models \cite{ccg2023,cc2013,cl2014,cl2015,glm2021,l2003,rbcps2018,sck2010,sgk2009,vabk2024}. However, it remains a challenging open problem to introduce the backstepping method to general multi-dimensional models. In contrast, controlling multi-dimensional PDEs, such as the heat equation or the wave equation, presents profoundly greater analytical and geometrical challenges. Control design becomes more difficult due to the complex interplay between the system's geometry, the spectrum of its spatial differential operator, and the actuator's location.


This work focuses on the multi-dimensional heat equation, which is a well-studied subject in the control theory. It is important to consider the seminal works of Lions~\cite{lions1971,lions1988}, who established a functional analysis framework for controllability and stabilization of distributed parameter systems. A step forward was achieved by Triggiani~\cite{triggiani1975,triggiani1980}, who proved abstract stabilizability results and boundary feedback stabilization for parabolic systems, via compact resolvent and spectral arguments. Subsequently, the works of Barbu and Triggiani~\cite{bt2004} and by Barbu and Wang~\cite{barbuwang2003} extended this theory to nonlinear and semilinear parabolic systems, establishing internal stabilization by finite-dimensional controllers. Breiten and Kunisch~\cite{breiten2014} proposed a Riccati feedback framework for reaction–diffusion systems arising in cardiac electrophysiology, illustrating the robustness of such designs in multi-dimensional domains. Recently, Badra and Takahashi~\cite{badra2014} presented the Fattorini criterion for approximate controllability and stabilization of parabolic systems, providing refined spectral characterizations.


In the multi-dimensional context, numerous studies have addressed systems subject to disturbances and developed diverse methods to counteract such disturbances, each with distinct advantages and applicability conditions. The choice of method depends on the nature of the disturbance considered and its structural relationship to the control input. When disturbance is a constant, it can be followed by traditional and well-known methods, which include but are not limited to the Spectral (Pole Placement) Method, the Riccati-Based Method, and the Backstepping Method (one-dimensional case) \cite{C2007,ks2008}. When the disturbance is not constant and time-dependent is more general and of significant interest. When the disturbance enters the system through the same channel as the control input (a \textit{matched} condition), robust methods like Sliding Mode Control (SMC) are highly effective. SMC drives the system trajectory onto a predetermined sliding manifold in finite time, inducing invariance to a class of matched disturbances \cite{bmp2025,gj2013,gl2014}. To handle the general and challenging case of bounded, unmatched disturbances, Active Disturbance Rejection Control (ADRC) offers a powerful solution. Its core principle is to treat all unknown dynamics and disturbances as a \textit{total disturbance}. This quantity is estimated in real-time by an Extended State Observer (ESO) and actively canceled by the control law \cite{zgp2025,zj2017}. This provides robustness without requiring a precise model of the disturbance itself.


It can also be mentioned the recent work of Balogoun, Marx, and Plestan~\cite{balogoun2025}, where well-posedness and global stabilization results for infinite-dimensional systems subject to disturbances and admissible control operators, using SMC, were established. Their control design relies on the sliding variable $\sigma(t) = \langle \phi, z(t) \rangle_H$, where $\phi$ is an eigenfunction of the adjoint operator $A_L^*$ associated with the closed-loop generator $A_L = A + BL$. In our setting, $A$ denotes the Laplacian and $B$ represents a localized internal control operator, we find the assumptions of~\cite{balogoun2025} are satisfied. However, the explicit computation of such eigenfunctions is considerably more involved. In contrast, our approach relies solely on the spectral properties of the Laplacian, thus avoiding the construction of $A_L$ and leading to a direct variational formulation of the feedback law. In the recent paper due to Labbadi and Roman~\cite{lr2025}, they achieved finite- and fixed-time stabilization by means of set-valued feedback of maximal monotone type. Although the localized control operator proposed here satisfies the same structural assumptions, their feedback involves additional nonlinear power terms. In contrast, the present method employs a simpler monotone law that preserves the natural dissipativity of the Laplacian. Finally, a relevant approach is found in Xiang~\cite{X2024}, where a localized finite-dimensional stabilizer for the multi-dimensional heat equation was built, via the Frequency Lyapunov method, a localized finite-dimensional stabilizer for the multi-dimensional heat equation using spectral arguments and well-chosen Lyapunov functions. That work is the starting point of our study.


The contribution of the present work is a novel control framework for the robust stabilization of a multi-dimensional heat equation subject to general bounded disturbances. It's designed a feedback control law that acts only on an arbitrarily small subdomain and incorporates a disturbance rejection mechanism to achieve robustness against a broad class of time- and space-varying unknown disturbances, ensuring a decay rate as large as desired. Providing a rigorous stability analysis for the resulting closed-loop system, establishing the well-posedness and exponential stability of the state to the desired equilibrium.  The stabilization of the multi-dimensional heat equation with localized disturbance has been previously studied in \cite{zj2017}, where only asymptotic stability via the ADRC approach was considered. While they get the eventual decay of the system's energy, the rate of convergence remains undetermined and potentially slow. In contrast, the present work yields a significantly stronger result, the exponential stability. Specifically, we prove the existence of constants $C\geq1$, such that for every $\lambda>0$ the system's energy $E(t)$ satisfies

\begin{equation*}
E(t) \leq C e^{-\lambda t} E(0), \quad \forall t>0 .    
\end{equation*}
\noindent This represents a qualitative improvement over asymptotic decay, as it provides a robust and fast rate of stabilization.


The remainder of this paper is organized as follows: Section \ref{Sec_P_S} formalizes the problem statement and presents the necessary mathematical preliminaries. Section \ref{Sec_Spec_Ineq} is devoted to the main tool of the present work, the Spectral inequality. The main stability theorem and its detailed proof are given in Section  \ref{Sec_Feedback}.  Section \ref{Sec_well_p} presents the well-posedness of the resulting closed-loop system via maximal monotone operator theory. Finally, Section 6 offers concluding remarks and directions for future research.



\section{Problem Statement}\label{Sec_P_S}
In the following, we present a precise formulation of the stabilization problem considered in this article.  Let $\Omega \subset \mathbb{R}^n$ ($n \in \mathbb{N}$)  be an open domain with smooth boundary $\partial \Omega$. Let $\omega \subset \Omega$ be a nonempty open subset of positive Lebesgue measure (i.e., $|\omega| > 0$). In this article, we focus our interest on a multi-dimensional heat equation controlled and perturbed on a subdomain:

\begin{equation}\label{P}\tag{P}
  \left\{\begin{array}{ll}
y_t- \Delta y=\chi_\omega(u+d),\quad&(t, x) \in(0, \infty) \times \Omega \\
y(t, x)=0,&(t, x) \in(0, \infty) \times \partial \Omega \\
y(0, x)=y_0(x), &x \in \Omega
\end{array}\right.
\end{equation}

\noindent where $\chi_\omega$ denotes the characteristic function on $\omega$, that is to say, $\chi_\omega(x)=1$ if $x \in \omega$ and $\chi_\omega(x)=0$ if $x \notin \omega$. The aim is to achieve exponential stabilization of system \eqref{P} by employing a distributed feedback control law $ u = u(t, x) $ that acts only on the interior subdomain $ \omega \subset \Omega $. The control must simultaneously suppress the effects of an unknown distributed disturbance $ d = d(t, x) $.

\noindent \noindent Regarding the undisturbed case $(d=0)$ the problem under consideration has been solved in Xiang \cite{X2024}.

\noindent The stabilization problem for partial differential equations subjected to unknown disturbances, acting either in the domain or at the boundary, has been object of recent interest. In Table \ref{literature} we present, without being exhaustive, some of the concerned literature.

\vspace{0.25cm}

\begin{table}[H]
\centering
\begin{tabular}{c|c|c|c}
Equation & Distributed disturbance & Boundary disturbance & Multidimensional \\ 
\hline
Heat & \cite{zgp2025}& \cite{hcp2020,gh2023,gh2025,zgp2025} &\cite{gp2020,X2024,zj2017,zgp2025} \\
Wave & \cite{fx2015,opu2011} & \cite{gj2013,gk2012,mgk2019,m2022,zw2018} &\cite{gz2024,h1989,lt1992,t1998,vfp2023}
\\
Beam & \cite{a2023,g2018} & \cite{gk2014} & --\\
\end{tabular}
\caption{\label{literature}Stabilization of partial differential equations subjected to unknown disturbances.}
\label{tl}
\end{table}

\vspace{0.25cm}



\noindent Throughout this work, we adopt standard notation for Sobolev and Hilbert spaces. We denote by $L^2(\Omega)$ the usual Hilbert space of square-integrable functions on $\Omega$, equipped with the inner product

\begin{equation*}
(u, v)_{L^2(\Omega)} = \int_\Omega u\, v\, ~dx,
\quad \text{and norm} \quad
\| u \|_{L^2(\Omega)}^2 = (u,u)_{L^2(\Omega)}.    
\end{equation*}

\noindent For $k \in \mathbb{N}$, the Sobolev space $H^k(\Omega)$ denotes the space of functions in $L^2(\Omega)$ with weak derivatives up to order $k$ in $L^2(\Omega)$, endowed with the norm

\begin{equation*}
\| u \|_{H^k(\Omega)}^2 
= \sum_{|\alpha| \leq k} \| D^\alpha u \|_{L^2(\Omega)}^2.    
\end{equation*}

\noindent We write $H_0^1(\Omega)$ for the subspace of $H^1(\Omega)$ consisting of functions with zero trace on $\partial \Omega$.

\noindent We introduce the second-order elliptic operator $\mathscr{A}$ given by

\begin{equation}\label{Op}
    \begin{cases}
\mathscr{A}: D(\mathscr{A}) \subset L^{2}(\Omega) \rightarrow L^{2}(\Omega)  \tag{Op}\\
D(\mathscr{A})=\left\{\phi \in H_{0}^{1}(\Omega)~ /~ \Delta \phi \in L^{2}(\Omega)\right\} \\
\mathscr{A} \phi= -\Delta \phi .
\end{cases}
\end{equation}

\noindent We note that $\mathscr{A}$ is selfadjoint and has compact resolvent. Hence, the spectrum of $\mathscr{A}$ consists of only isolated eigenvalues with finite multiplicity. Furthermore, there exists a Hilbert orthogonal basis $\left\{e_i\right\}_{i\in \mathbb{N}}$ of $D(\mathscr{A})$ consisting of eigenfunctions of $\mathscr{A}$, associated with the sequence of eigenvalues $\left\{\tau_i\right\}_{i \in \mathbb{N}}$. Note that

\begin{equation}\label{decomp_espec}
    \begin{aligned}
   & 0<\tau_1 \leq \tau_2 \leq \tau_3 \leq \ldots \leq \tau_i \leq \ldots <+\infty\quad \text { and } \quad \tau_i \underset{i \rightarrow+\infty}{\longrightarrow}\infty.\\
&-\Delta e_i=\tau_i e_i \text { with }\left.e_i\right|_{\partial \Omega}=0.
\end{aligned}
\end{equation}

\noindent Different eigenvalues $\tau_i$ may coincide, but each eigenvalue only has finite algebraic multiplicity.

\noindent Given $\lambda > 0$, let $N(\lambda) := \# \{\, i \in \mathbb{N} : \tau_i \leq \lambda \},$ i.e., $\tau_{N(\lambda)}\leq\lambda<\tau_{N(\lambda)+1}$. Then Weyl’s law \cite{weyl} gives $N(\lambda) \sim (2\pi)^{-n} \omega_n\, \operatorname{vol}(\Omega)\, \lambda^{n/2}$, where $\omega_n$ is the volume of the unit ball in $\mathbb{R}^n$. Henceforth, for the sake of notational convenience, we shall denote $N(\lambda)$ simply by $N$. 

\noindent For any $y \in L^2(\Omega)$, its eigenfunction expansion reads

\begin{equation*}
 y(x) = \sum_{i=1}^\infty y_i e_i(x), 
\quad 
y_i := (y, e_i)_{L^2(\Omega)}.   
\end{equation*}

\noindent  We define the orthogonal projection onto the span of the first $N$ eigenfunctions by $P_N y:= \displaystyle\sum_{i=1}^N y_i e_i,$ and $P_N^\perp$ the co-projection.

\noindent In this setting, both the control input and the disturbance are expanded in the eigenfunction basis. Specifically, the control and disturbance terms localized to the control region $\omega$ take the form

\begin{equation*}
\chi_\omega\, d(t,x) = \chi_\omega \sum_{i=1}^{\infty} d_i(t)\, e_i(x), 
\quad
\chi_\omega\, u(t,x) =\chi_\omega \sum_{i=1}^{\infty} u_i(t)\, e_i(x).    
\end{equation*}



\noindent Although the disturbance is assumed to be unknown, we ask it to satisfy the following two assumptions, which are the standard ones that can be found in the literature.


\begin{enumerate}[label=\textbf{(A\arabic*)}]
    \item\label{A_1} $d \in L^{1}\left(0, \infty ; L^{2}(\Omega)\right)$.
    \item\label{A_2} There exists $D \in(0, \infty)$ such that $\|d(t, \cdot)\|_{L^{2}(\Omega)} \leq D$, for every $t\in[0,\infty)$.
\end{enumerate}

\noindent In order to reject the effects of the disturbance, we use the sign multivalued operator in a Hilbert space ${H}$, $\operatorname{sign}_{{H}}(\cdot): {H} \rightarrow 2^{{H}}\left(2^{H}\right.$ denotes the power set of $\left.H\right)$, given by

\begin{equation}\label{sign}
\operatorname{sign}_{{H}}(f)=\begin{cases}
\dfrac{f}{\|f\|_{H}},& \text { if } f \neq 0  \tag{sign}\\[2ex]
\left\{g \in {H} /\|g\|_{{H}} \leq 1\right\} ,&\text { if } f=0.
\end{cases}  
\end{equation}

\noindent To that end, we employ the property of the multivalued sign operator in the Hilbert space $L^{2}(\Omega)$.

\begin{equation*}
\int_{\Omega} \theta f ~d x=\|f\|_{L^2(\Omega)}, \quad\forall f \in L^2(\Omega),  \quad\forall \theta \in \operatorname{sign}_{L^2(\Omega)}(f) .    
\end{equation*}

\noindent  Our main result is the following one:

\begin{thm}\label{main_theorem}
Let us assume \ref{A_1} and \ref{A_2}. Let $y_{0}$ in $L^{2}(\Omega)$ be the initial condition. Let $\lambda$ in $(0, \infty)$ be the desired decay rate. Then, there exists a  feedback law $\mathscr{G_\lambda}: L^{2}(\Omega) \rightarrow$ $L^{2}(\Omega)$ such that \eqref{P} is exponentially stable in $L^{2}(\Omega)$, with decay rate $\lambda$. Being more explicit, \eqref{P} with the feedback law $u=\mathscr{G_\lambda}$ has a unique weak solution $y$ in $C\left([0, \infty) ; L^{2}(\Omega)\right)$, and it satisfies

\begin{equation}\label{eq:decay_solution}
\|y(t, \cdot)\|_{L^{2}(\Omega)} \leq e^{-\lambda t}\left\|y_{0}\right\|_{L^{2}(\Omega)}, t \in[0, \infty) 
\end{equation}

\begin{rem}
Assumption \ref{A_1} is required for the well-posedness part of Theorem 1, and assumption \ref{A_2} is needed for the construction of the feedback law.    
\end{rem} 

\end{thm}



\section{Spectral Inequality}\label{Sec_Spec_Ineq}

A crucial tool in the feedback design is the spectral inequality, which ensures that the modal energy in the controlled region is sufficiently observable. For instance, concerning the eigenfunctions $\left\{e_i\right\}_{i \in \mathbb{N}}$ one has the following result:
\begin{prop}
    The eigenfunctions $\left\{e_i\right\}_{i \in \mathbb{N}}$ satisfy
    \begin{enumerate}
        \item  Orthonormal basis: $\left(e_i, e_j\right)_{L^2(\Omega)}=\delta_{i j}$.
        \item (Unique continuation) The symmetric matrix $J_N$ given below is invertible \cite{bt2004},

\begin{equation*}
J_N:=\left(\left(e_i, e_j\right)_{L^2(\omega)}\right)_{i, j=1}^N  .  
\end{equation*}

        \item (Weak Spectral inequality) There exist $C_\lambda > 0$ such that for $Y_{N}=\left(a_1, \ldots, a_{N}\right) \in \mathbb{R}^N$, we have

\begin{equation}\label{spect_ineq}
 Y_{N}^{T} J_{N} Y_{N(x)}=\left\|\sum_{n=1}^{N} a_n e_n\right\|_{L^2(\omega)}^2 \geq C_\lambda \sum_{n=1}^{N} a_n^2 \text {. }   
\end{equation}
    \end{enumerate}

\end{prop}

\begin{rem}
    The weak spectral inequality follows as a consequence of the unique continuation property.
\end{rem}

\begin{rem}
    For the Laplacian operator, a more precise spectral inequality - known as the \emph{quantitative spectral inequality} - was introduced by Lebeau and Robbiano \cite{lr1995}, and previously discussed in \cite{lin1991}. This inequality was later used by Xiang \cite{X2024} to establish quantitative rapid stabilization results. In this work, we employ the weak spectral inequality instead, as our focus is on rapid stabilization rather than its quantitative version. As a result, the constant in our spectral inequality does not depend on the decay rate parameter $\lambda$.
\end{rem}


\section{Feedback Design}\label{Sec_Feedback}


A main objective in controlling distributed parameter systems is the development of feedback laws which ensure the system state converges to a prescribed target while maintaining robustness in the presence of external disturbances. In this section, we prove our main result related with the exponential stabilization of the heat equation with localized control and bounded perturbation. We base our ideas on \cite{X2024}. We present the construction of such a feedback law for problem \eqref{P}, assuming suitable regularity of the solution. To clarify the main ideas underlying the feedback design, we decompose the control input $u$ as

\begin{equation*}
u = \tilde{u} + \hat{u},    
\end{equation*}

\noindent where the term $\tilde{u}$ is designed to achieve the desired decay rate of the state, while $\hat{u}$ will be designed separately to mitigate the impact of the disturbance, $d$, while preserving the stabilizing effect of $\tilde{u}$.

\noindent For completeness, we recall several useful identities that will be used in the subsequent analysis. First, the projection of the eigenfunction $e_j$ onto the control region is given by

\begin{equation*}
\chi_\omega e_j 
= \sum_{i=1}^{\infty} (\chi_\omega e_j,\, e_i)_{L^2(\Omega)}\, e_i 
= \sum_{i=1}^{\infty} (e_i,\, e_j)_{L^2(\omega)}\, e_i.  
\end{equation*}

\noindent The time derivative and Laplacian terms satisfy the standard orthogonality properties:

\begin{equation*}
\int_{\Omega} y_t\, e_j~ dx 
= \sum_{i=1}^{\infty} \frac{d}{dt} y_i\, \int_{\Omega} e_i\, e_j~ dx 
= \frac{d}{dt} y_j,    
\end{equation*}

\begin{equation*}
\begin{aligned}
    \int_{\Omega} \Delta y e_{j} ~d x&=\int_{\Omega} \sum_{i=1}^{\infty} y_{i}\left(\Delta e_{i}, e_{j}\right)~ d x  =-\sum_{i=1}^{\infty} y_{i} \tau_{i} \int_{\Omega} e_{i} e_{j} ~d x 
& =-\tau_{j} y_{j} ,
\end{aligned}    
\end{equation*}

\noindent where $\tau_j$ is the eigenvalue associated with  $e_j$.

\noindent Finally, the projections of the feedback and disturbance components are given by

\begin{equation*}
\int_{\Omega} \chi_\omega 
\Bigg( \sum_{i=1}^{N} e_i\, \tilde{u}_i(t) \Bigg) e_j~ dx 
= \sum_{i=1}^{N} \tilde{u}_i(t)\, (e_i,\, e_j)_{L^2(\omega)},    
\end{equation*}

\begin{equation*}
\int_{\Omega} \chi_\omega 
\Bigg( \sum_{i=1}^{N} e_i\, [\,\hat{u}_i(t) + d_i(t)] \Bigg) e_j~ dx 
= \sum_{i=1}^{N} [\,\hat{u}_i(t) + d_i(t)]\, (e_i,\, e_j)_{L^2(\omega)}.    
\end{equation*}

\noindent Then, we deduce

\begin{equation}
\begin{cases}\displaystyle 
y_1^{\prime}(t)  =  -\tau_1 y_1(t)+\sum_{i=1}^{N} \tilde{u}_{i}(t)\left(e_{i}, e_{1}\right)_{L^{2}(\omega)}+\sum_{i=1}^{N} (\hat{u}_{i}(t)+d_i(t)) \left(e_{i}, e_{1}\right)_{L^{2}(\omega)} \\
\quad\vdots \qquad\qquad\vdots\hfill\vdots\qquad\qquad\hfill\qquad\vdots \hfill\\\displaystyle 
y_N^{\prime}(t) = -\tau_N y_N(t)+\sum_{i=1}^{N} \tilde{u}_{i}(t)\left(e_{i}, e_{N}\right)_{L^{2}(\omega)}+\sum_{i=1}^{N} (\hat{u}_{i}(t)+d_i(t)) \left(e_{i}, e_{N}\right)_{L^{2}(\omega)} .
\end{cases}
\end{equation}

\noindent Thus, with the aid of the
matrices

\begin{equation*}
X_N(t):=\begin{pmatrix}
y_1(t) \\
y_2(t) \\
\vdots \\
y_N(t)
\end{pmatrix}, \quad U_N(t):=\begin{pmatrix}
\tilde{u}_1(t) \\
\tilde{u}_2(t) \\
\vdots \\
\tilde{u}_N(t)
\end{pmatrix}, \quad A_N:=\begin{pmatrix}
-\tau_1 & & & \\
& -\tau_2 & & \\
& & \ddots & \\
& & & -\tau_N
\end{pmatrix}
\end{equation*}
\begin{equation*}
   F_N :=\begin{pmatrix}
\hat{u}_1(t)+d_1(t) \\
\hat{u}_2(t)+d_2(t) \\
\vdots \\
\hat{u}_N(t)+d_N(t)
\end{pmatrix},
\end{equation*} and using the definition of $J_N$, we  can construct a finite system 
 
 \begin{equation*}
 \dot{X}_N(t)=A_N X_N(t)+J_N U_N(t)+  J_N  F_N . 
 \end{equation*}


\subsection{Design of \texorpdfstring{$\tilde{u}$}{u}}\label{subsec2-1}

\noindent \noindent In this section, we present the construction of the feedback control law and introduce a Lyapunov function designed for stability analysis. For a given parameter $\lambda$ (and consequently, a fixed $N$), we select constants $\gamma_\lambda, \mu_\lambda > 0$, which will be specified later. Following the approach suggested in \cite{X2024}, we define the feedback control law as:  

\begin{equation*}
U_N(y(t)) := -\gamma_\lambda X_N(t).
\end{equation*}  

\noindent To analyze the stability of the closed-loop system, we introduce the following Lyapunov function, referred to as the Frequency Lyapunov function  

\begin{equation*}
V(y) := \mu_\lambda \|X_N\|^2 + \|P_N^\perp y\|^2_{L^2(\Omega)}, \quad \forall y \in L^2(\Omega).
\end{equation*}  

\noindent Here, $\|X_N\|^2$ denotes the Euclidean norm $X_N^T X_N = \sum_{i=1}^N y_i^2$, which is equivalent to $\|P_N y\|_{L^2(\Omega)}^2$.

\noindent For any initial state $y_0 \in H_0^1(\Omega)$, the time derivative of $V(y(t))$ is computed as follows 

\begin{equation*}
\begin{aligned}  
\frac{d}{dt} V(y(t)) &= \frac{d}{dt} \left( \mu_\lambda \|X_N\|^2 \right) + \frac{d}{dt} \left( P_N^\perp y, P_N^\perp y \right)_{L^2(\Omega)} \\[2ex]  
&= \mu_\lambda \frac{d}{dt} \sum_{i=1}^N y_i^2 + \frac{d}{dt} \left( \sum_{i=N+1}^\infty y_i e_i, \sum_{i=N+1}^\infty y_i e_i \right)_{L^2(\Omega)} \\[2ex]  
&= \mu_\lambda \frac{d}{dt} \|X_N\|^2 + 2 \left\langle P_N^\perp y, \frac{d}{dt} y \right\rangle_{H_0^1(\Omega), H^{-1}(\Omega)}.  
\end{aligned}
\end{equation*}  

\noindent Expanding the derivative of $\|X_N\|^2$ further, we obtain

\begin{equation*}
\begin{aligned}  
\mu_\lambda \frac{d}{dt} \|X_N\|^2 &= \mu_\lambda \left( \dot{X}_N^\top X_N + X_N^\top \dot{X}_N \right) \\[2ex] 
&=\mu_{N}\left[\left(A_{N} X_{N}-\gamma_{\lambda} J_{N} X_{N}+J_{N}F_{N}\right)^{\top} X_{N}+X_{N}^{\top}\left(A_{N} X_{N}-\gamma_{\lambda} J_{N} X_{N}+J_{N}F_{N}\right)\right]\\[2ex]
&= \mu_\lambda \left[ X_N^\top \left( 2A_N - 2\gamma_\lambda J_N \right) X_N + 2F_N^\top J_N X_N \right] \\[2ex]  
&= 2\mu_\lambda X_N^\top A_N X_N - 2\mu_\lambda \gamma_\lambda X_N^\top J_N X_N + 2\mu_\lambda F_N^\top J_N X_N.  
\end{aligned}
\end{equation*}

\noindent By using \eqref{spect_ineq}, we have

\begin{equation*}
    \begin{aligned}
       \mu_{\lambda} \frac{d}{d t}\left\|X_{N}\right\|_{2}^{2} &\leq -2 \mu_{\lambda}\gamma_{\lambda} X_{N}^{\top}  J_{N} X_{N}+2\mu_{\lambda}F_{N}^{\top}  J_{N} X_{N} \\[2ex] 
& \leq - 2\mu_\lambda \gamma_\lambda C_\lambda\left\|X_N\right\|^2+2 \mu_{\lambda} \sum_{j=1}^{N}\sum_{i=1}^{N}\left(\hat{u}_i(t)+d_i(t)\right)\left(e_{i}, e_{j}\right)_{L^{2}(\omega)}  y_{j}(t)\\[2ex] 
&=- 2\mu_\lambda \gamma_\lambda C_\lambda\left\|X_N\right\|^2+2\mu_\lambda\left((\hat{u}+d),\chi_\omega P_N y\right)_{L^2(\Omega)} .
    \end{aligned}
\end{equation*}

\noindent On the other hand, by using the Cauchy–Schwarz inequality, weak spectral inequality \eqref{decomp_espec}, and Cauchy-$\varepsilon$ inequality with $\varepsilon = \frac{\gamma_\lambda}{\lambda} > 0$, it follows

\begin{equation*}
    \begin{aligned}
& 2\left\langle P_N^\perp y, \frac{d}{d t} y\right\rangle_{H_0^1(\Omega) \times H^{-1}(\Omega)}\\[2ex]
&=2\left\langle P_N^\perp y, \Delta y-\gamma_\lambda \chi_\omega\left(\sum_{i=1}^{N} e_i y_i(t)\right)+\chi_\omega(\hat{u}+d)\right\rangle_{H_0^1(\Omega) \times H^{-1}(\Omega)} \\[2ex]
& =2\left( P_N^\perp y, \Delta y\right)_{L^2(\Omega)}-2 \gamma_\lambda\left( P_N^\perp y, \chi_\omega\left(\sum_{i=1}^N e_i y_i\right)\right)_{L^2(\Omega)}+2\left(P_N^\perp y,\chi_\omega(\hat{u}+d)\right)_{L^2(\Omega)} \\
& =2 \sum_{j=N+1}^{\infty} y_j^2\left(e_j, \Delta e_i\right)_{L^2(\Omega)}-2 \gamma_\lambda\left(P_N^\perp y, \chi_\omega P_N y\right)_{L^2(\Omega)}+2\left(P_N^\perp y,\chi_\omega(\hat{u}+d)\right)_{L^2(\Omega)}  \\[2ex]
& \leq -2 \sum_{j=N+1}^{\infty} y_j^2 \tau_j+2 \gamma_\lambda\left(P_N^\perp y, P_N y\right)_{L^2(\omega)} +2\left(P_N^\perp y,\chi_\omega(\hat{u}+d)\right)_{L^2(\Omega)} \\[2ex]
&\leq-2 \sum_{j=N+1}^{\infty} \lambda y_j^2+2 \gamma_\lambda\left\|P_N^{\perp}y\right\|_{L^2(\omega)}\left\| P_N y\right\|_{L^2(\omega)}+2\left(P_N^\perp y,\chi_\omega(\hat{u}+d)\right)_{L^2(\Omega)} \\[2ex]
& \leq-2 \lambda\left\|P_{N}^{\perp} y\right\|_{L^{2}(\Omega)}^{2}+\lambda\left\|P_{N}^{\perp} y\right\|_{L^{2}(\Omega)}^{2}+\frac{\gamma_{\lambda}^{2}}{\lambda}\left\|P_{N} y\right\|_{L^{2}(\Omega)}^{2}+2\left(\chi_\omega P_N^\perp y,(\hat{u}+d)\right)_{L^2(\Omega)}  \\[2ex]
& =-\lambda\left\|P_{N}^\perp y\right\|_{L^2(\Omega)}^{2}+\frac{\gamma_{\lambda}^{2}}{\lambda}\left\|X_{N} \right\|^{2}+2\left((\hat{u}+d),\chi_\omega P_N^\perp y\right)_{L^2(\Omega)}.
\end{aligned}
\end{equation*}
\noindent Therefore,

\begin{equation*}
\begin{aligned}
    \frac{d}{dt} V(y(t))\leq& \left(- 2\mu_\lambda \gamma_\lambda C_\lambda+\frac{\gamma_{\lambda}^{2}}{\lambda} \right)\left\|X_{N} \right\|^{2}-\lambda\left\|P_{N}^\perp y\right\|_{L^{2}(\Omega)}^2+2 \left((\hat{u}+d),\mu_{\lambda}\chi_\omega P_N y+\chi_\omega P_N^\perp y\right)_{L^2(\Omega)} .
    \end{aligned}
\end{equation*}


\subsection{Design of \texorpdfstring{$\hat{u}$}{u}}\label{subsec2-2}

To attenuate the effect of the unknown disturbance $ d(t,x) $, we now construct the control term $ \hat{u} $. Let $ {H} = L^2(\Omega) $, and recall the key property of the multivalued sign operator in $ L^2(\Omega) $

\begin{equation*}
\int_{\Omega} \theta f ~d x=\|f\|_{L^2(\Omega)}, \quad\forall f \in L^2(\Omega),  \quad\forall \theta \in \operatorname{sign}_{L^2(\Omega)}(f) .    
\end{equation*}

\noindent Since the disturbance is supposed to be unknown, we can not choose $\hat{u}(t,x)=-d(t,x)$.  Therefore, thanks to the Cauchy-Schwarz inequality, we have 

\begin{equation*}
\begin{aligned}
&2\left((\hat{u}+d),\mu_{\lambda}\chi_\omega P_N y+\chi_\omega P_N^\perp y\right)_{L^2(\Omega)}  \\[2ex]
&= 2\left(\hat{u},\chi_\omega (\mu_{\lambda}P_N y+\ P_N^\perp y)\right)_{L^2(\Omega)}+2\left(d,\chi_\omega (\mu_{\lambda}P_N y+\ P_N^\perp y)\right)_{L^2(\Omega)}\\[2ex]
&\leq \left(\hat{u},\chi_\omega (\mu_{\lambda}P_N y+\ P_N^\perp y)\right)_{L^2(\Omega)}+2\left\|d(t,\cdot)\right\|_{L^2(\Omega)}\left\|\chi_\omega (\mu_{\lambda}P_N y+\ P_N^\perp y)\right\|_{L^2(\Omega)}
    \end{aligned}
\end{equation*}

\noindent Under assumption \ref{A_2}, the disturbance satisfies $ \|d(t, \cdot)\|_{L^2(\Omega)} \leq D $. Then, it suffices to consider the feedback law

\begin{equation*}
    \hat{u}(t,x)=-D \operatorname{sign}_{L^2(\Omega)}\left[\chi_\omega (\mu_{\lambda}P_N y+\ P_N^\perp y)\right],
\end{equation*}

\noindent thereby guaranteeing that

\begin{equation*}
\begin{aligned}
&2\left((\hat{u}+d),\mu_{\lambda}\chi_\omega P_N y+\chi_\omega P_N^\perp y\right)_{L^2(\Omega)}  \\[2ex]
&\leq -2D\left(\operatorname{sign}_{L^2(\Omega)}\left[\chi_\omega (\mu_{\lambda}P_N y+\ P_N^\perp y)\right],\chi_\omega (\mu_{\lambda}P_N y+\ P_N^\perp y)\right)_{L^2(\Omega)}+2D\left\|\chi_\omega (\mu_{\lambda}P_N y+\ P_N^\perp y)\right\|_{L^2(\Omega)}\\[2ex]
&\leq -2D\left\|\chi_\omega (\mu_{\lambda}P_N y+\ P_N^\perp y)\right\|_{L^2(\Omega)}+2D\left\|\chi_\omega (\mu_{\lambda}P_N y+\ P_N^\perp y)\right\|_{L^2(\Omega)}=0.
    \end{aligned}
\end{equation*}

\noindent From the previous derivations, we obtain

\begin{equation}\label{eq:derivative_initial}
    \frac{d}{dt} V(y(t))\leq  \left(- 2\mu_\lambda \gamma_\lambda C_\lambda+\frac{\gamma_{\lambda}^{2}}{\lambda} \right)\left\|X_{N} \right\|^{2}-\lambda\left\|P_{N}^\perp y\right\|_{L^{2}(\Omega)}^2.
\end{equation}

\noindent We now fix the parameters

\begin{equation*}
\gamma_\lambda := \frac{\lambda}{C_\lambda}, \qquad \mu_\lambda := \frac{\gamma_\lambda^2}{\lambda^2} = C_\lambda^{-2},   
\end{equation*}

\noindent which, when substituted into \eqref{eq:derivative_initial}, yield

\begin{equation*}
\begin{aligned}
\frac{d}{dt} V(y(t)) 
&\leq \left( - 2\mu_\lambda \lambda + \mu_\lambda \lambda \right) \|X_N\|^2 - \lambda \|P_N^\perp y\|_{L^2(\Omega)}^2 \\[2ex]
&= -\mu_\lambda \lambda \|X_N\|^2 - \lambda \|P_N^\perp y\|_{L^2(\Omega)}^2 \\[2ex]
&= -\lambda \left( \mu_\lambda \|X_N\|^2 + \|P_N^\perp y\|_{L^2(\Omega)}^2 \right) = -\lambda V(y(t)).
\end{aligned}  
\end{equation*}

\noindent This differential inequality implies the exponential decay of $V(y(t))$. Indeed, we observe that

\begin{equation*}
\frac{d}{dt} \left( e^{\lambda t} V(y(t)) \right) = e^{\lambda t} \left( \frac{d}{dt} V(y(t)) + \lambda V(y(t)) \right) \leq 0.   
\end{equation*}

\noindent Upon integration, this becomes

\begin{equation*}
 V(y(t)) \leq e^{-\lambda t} V(y(0)).   
\end{equation*}

\noindent Let us introduce the constants

\begin{equation*}
 \alpha := \min\left\{1, \sqrt{\mu_\lambda}\right\} = \min\left\{1, \frac{1}{C_\lambda} \right\}, \qquad
\beta := \max\left\{1, \sqrt{\mu_\lambda}\right\} = \max\left\{1, \frac{1 }{C_\lambda} \right\}.   
\end{equation*}

\noindent Using these constants, we can bound the Lyapunov functional as follows

\begin{equation*}
\begin{aligned}
V(y(t)) &= \mu_\lambda \|P_N y(t)\|_{L^2(\Omega)}^2 + \|P_N^\perp y(t)\|_{L^2(\Omega)}^2 \\[2ex]
&\geq \alpha^2 \left( \|P_N y(t)\|_{L^2(\Omega)}^2 + \|P_N^\perp y(t)\|_{L^2(\Omega)}^2 \right) = \alpha^2 \|y(t)\|_{L^2(\Omega)}^2,
\end{aligned}
\end{equation*}

\noindent and similarly,

\begin{equation*}
\begin{aligned}
V(y(0)) & \leq\left(\beta^2\left\|P_N y(0)\right\|_{L^2(\Omega)}^2+\beta^2\left\|P_N^\perp y(0)\right\|_{L^2(\Omega)}^2\right) =\beta^2\|y(0)\|_{L^2(\Omega)}^2.
\end{aligned}  
\end{equation*}

\noindent Combining these bounds with the exponential decay of $V(y(t))$, we conclude that

\begin{equation*}
\alpha^2 \|y(t)\|_{L^2(\Omega)}^2 \leq V(y(t)) \leq e^{-\lambda t} V(y(0)) \leq \beta^2 e^{-\lambda t} \|y(0)\|_{L^2(\Omega)}^2.    
\end{equation*}

\noindent Taking square roots, we obtain the exponential stability estimate

\begin{equation*}
\|y(t)\|_{L^2(\Omega)} \leq \frac{\beta}{\alpha} e^{-\frac{\lambda}{2} t} \|y(0)\|_{L^2(\Omega)}.    
\end{equation*}


\subsection{Feedback law}\label{feedback_design}

\noindent The goal of this section is to derive the feedback law. We begin by constructing the associated feedback operator. The feedback control system is defined through the following operators:
\begin{enumerate}
\item The \emph{linear feedback operator} $\mathscr{F}: L^2(\Omega) \to L^2(\Omega)$:

\begin{equation*}\label{eq:F_operator}
\mathscr{F}\phi=-\gamma_\lambda\left(\sum_{m=1}^{N}\left(\phi,e_i\right)_{L^2(\Omega)} e_i\right)=-\gamma_\lambda P_{N} \phi, ~\text{with } \gamma_\lambda=C_\lambda^{-1}\lambda.
\end{equation*}

\item The \emph{nonlinear set-valued operator} $\mathscr{B}: L^2(\Omega) \to 2^{L^2(\Omega)}$:

\begin{equation*}\label{eq:B_operator}
\mathscr{B} \phi := -D  \operatorname{sign}_{L^2(\Omega)}(\chi_\omega \mathscr{C} \phi), \quad \text{ where } \mathscr{C} y := \mu_\lambda P_N y + P_N^\perp y.
\end{equation*}

\end{enumerate}

\noindent Finally, recalling that the total control is given by the decomposition $u=\tilde{u}+\hat{u}$, given by

\begin{equation*}
\tilde{u}(t, x)=\mathscr{F}(y(t, x)) ,\quad \hat{u}(t,x)= \mathscr{B}(y(t,x))  ,
\end{equation*}
\noindent we obtain the explicit form of the feedback law

\begin{equation}\label{feedback}
\begin{aligned}
u(t, x) & = -\gamma_\lambda  P_{N} y -D \operatorname{sign}_{L^2(\Omega)}(\chi_\omega \mathscr{C} y).
\end{aligned}
\end{equation}

\noindent Then, the closed-loop System is given by  

\begin{equation}\label{eq:closed_loop_system}
\begin{cases}
\partial_t y + \mathscr{A}y - \chi_\omega \mathscr{F} y - \chi_\omega \mathscr{B} y\ni \chi_\omega d(t) & \text{in } (0,\infty) \times \Omega \\
y = 0 & \text{on } (0,\infty) \times \partial\Omega \\
y(0,x) = y_0(x) & \text{in } \Omega.
\end{cases}
\end{equation}



\section{Well-Posedness}\label{Sec_well_p}

In this section, we present the well-posedness of the closed-loop system \eqref{eq:closed_loop_system}, which is a differential inclusion, through maximal monotone operator theory. To this end, we  reformulate the system in an abstract setting by introducing the multivalued operator

\begin{equation}\label{operator}
\begin{cases}
A: D({A}) \subset L^2(\Omega) \longrightarrow 2^{L^2(\Omega)} \\
D(A) = D(\mathscr{A})\\
A(y)= -\Delta y + \gamma_\lambda\chi_\omega P_N y  + \chi_\omega D \operatorname{sign}_{L^2(\Omega)}(\chi_\omega \mathscr{C}y).
\end{cases}
\end{equation}

\noindent It follows that the resulting closed-loop system is the differential inclusion

\begin{equation}\label{system} 
 \begin{cases}
y^{\prime}(t)+A y(t) \ni\chi_\omega d(t), \quad t>0 \\
y(0)=y_0.
\end{cases}   
\end{equation}

\noindent We define the inner product $ \langle\cdot, \cdot\rangle_{\mu} $ on $ L^2(\Omega) $ by

\begin{equation*}
(u, v)_{\mu} = \mu (P_N u, P_N v) + (P_N^\perp u, P_N^\perp v).
\end{equation*}

\noindent This bilinear form induces a norm $ \|u\|^2_\mu = (u, u)_\mu =  \mu \|P_N u\|^2 + \|P_N^\perp u\|^2  $, which is equivalent to the original norm on $ L^2(\Omega) $. Since $\mu>1$, we have $\|u\|_{L^2(\Omega)}^2 \leq(u, u)_\mu \leq \mu\|u\|_{L^2(\Omega)}^2$, it follows that $(L^2(\Omega), \langle\cdot,\cdot\rangle_{L^2(\Omega)})$ and $(L^2(\Omega),\langle\cdot,\cdot\rangle_\mu)$ define the same topology.  Moreover, notice that $ V(y(t))=\|u\|^2_\mu = \mu ( u, P_N v) + ( u, P_N^\perp v)=(u,\mathscr{C} v) $, which makes natural the choice of the inner product and will play a central role in this section.

\noindent To establish the well-posedness of the system, we will employ the maximal monotone operator theory. In this framework, we present two key results: the first one (Proposition \ref{wp1}) states that the operator $A$ is monotone, while the second one (Proposition \ref{wp2}) states that the operator $I+A$ is surjective.


\begin{prop}
\label{wp1}
\noindent The operator $A$ defined by \eqref{operator} is monotone.
\end{prop}


\begin{proof}

\noindent Let $y_1,y_2\in  D({A})$,  there exists $\eta_i\in \operatorname{sign}_{L^{2}(\Omega)}(\chi_\omega \mathscr{C} y_i) $, $i=1,2,$. Note that we can decompose $A$ into linear and nonlinear parts. Its linear part is given by 

\begin{equation*}
    A_1y:=-\Delta y + \gamma_\lambda\chi_\omega P_N y,
\end{equation*}
\noindent and the nonlinear part is

\begin{equation}
\label{eq: B nonlinear}
    B y=\chi_\omega D \operatorname{sign}_{L^2(\Omega)}(\chi_\omega \mathscr{C}y).
\end{equation}
\noindent We want to prove 

\begin{equation*}
  \left\langle Ay_1-Ay_2,y_1-y_2 \right\rangle_\mu=\langle A_1 y_1 -A_1 y_2 ,y_1-y_2\rangle_\mu+ \langle B y_1 +By_2, y_1-y_2\rangle_\mu\geq 0.  
\end{equation*}

\noindent To analyze the linear part, we consider $z=y_1-y_2$. This give us

\begin{equation}\label{eq:monotonicity_A1}
\begin{aligned} 
\left\langle A_1 z,z\right\rangle_{\mu} =& -\mu_\lambda\left\langle \Delta z, P_N z \right\rangle_{L^{2}(\Omega)} -\left\langle\Delta z , P_N^\perp z \right\rangle_{L^{2}(\Omega)} + \mu_\lambda\gamma_{\lambda} \left\langle \chi_{\omega} P_N z,P_N z \right\rangle_{L^{2}(\Omega)}\\
&+  \gamma_{\lambda} \left\langle \chi_{\omega} P_N z,P_N^\perp z \right\rangle_{L^{2}(\Omega)}.
\end{aligned}    
\end{equation}

\noindent Notice that, by using integration by parts in the first and second term of \eqref{eq:monotonicity_A1}, we get

\begin{align} 
&-\mu_\lambda\left\langle \Delta z, P_N z \right\rangle_{L^2 (\Omega) } -\left\langle\Delta z , P_N^\perp z \right\rangle_{L^2(\Omega)}\notag \\[2ex]
&= \mu_\lambda\sum_{i=1}^N z_i^2(\nabla e_i, \nabla e_i)_{L^2(\Omega)^n} -\frac{1}{2}\sum^\infty_{i=N+1}y_i^2(\Delta e_i, e_i)_{L^2(\Omega)} +\frac{1}{2}\sum^\infty_{i=N+1}z_i^2(\nabla e_i, \nabla e_i)_{L^2(\Omega)^n} \notag  \\[2ex]
&= \mu_\lambda\|\nabla (P_N z)\|^2_{L^2(\Omega)}+\frac{1}{2} \sum_{i=N+1}^{\infty}{\tau}_i z_i^2+\|\nabla(P_N^\perp z)\|^2_{L^2(\Omega)}\notag\\[2ex]
&\geq \mu_\lambda\|\nabla (P_N z)\|^2_{L^2(\Omega)}+\frac{1}{2}\|\nabla(P_N^\perp z)\|^2_{L^2(\Omega)}+\frac{\lambda}{2}\|P_N^\perp z\|^2_{L^2(\Omega)}\notag\\[2ex]
&\geq\frac{1}{2}\|\nabla z\|^2_{L^2(\Omega)}+\frac{\lambda}{2}\|P_N^\perp z\|^2_{L^2(\Omega)}
\label{eq:monotonicity_A1_P1}.
\end{align}

\noindent Furthermore, using \eqref{spect_ineq} and the Cauchy-$\varepsilon$ inequality in the third and fourth term of \eqref{eq:monotonicity_A1}, taking $\varepsilon = \frac{\sqrt{2}}{ C_\lambda}$,  it yields

\begin{align}
&\mu_\lambda\gamma_{\lambda} \left\langle \chi_{\omega} P_N z,P_N z \right\rangle_{L^{2}(\Omega)}\notag\\[2ex]
&\geq \mu_\lambda\gamma_{\lambda}\|P_N y\|^2_{L^2(\omega)}-\gamma_\lambda\|P_N y\|_{L^2(\Omega)}\|P_N^\perp y\|_{L^2(\omega)}\notag\\[2ex]
&\geq \mu_\lambda\gamma_{\lambda} C_\lambda\|P_N y\|^2_{L^2(\Omega)}-\frac{\gamma_\lambda\varepsilon}{2}\|P_N y\|^2_{L^2(\Omega)}-\frac{\gamma_\lambda}{2\varepsilon}\|P_N^\perp y\|^2_{L^2(\Omega)}\notag\\[2ex]
&=\left(\mu_\lambda\gamma_{\lambda} C_\lambda-\frac{\gamma_\lambda\sqrt{2}}{2 C_\lambda}\right)\|P_N y\|^2_{L^2(\Omega)}-\frac{\gamma_\lambda C_\lambda}{2\sqrt{2}}\|P_N^\perp y\|^2_{L^2(\Omega)}\label{eq:monotonicity_A1_P2}.
\end{align}

\noindent Now, using \eqref{eq:monotonicity_A1_P1} and \eqref{eq:monotonicity_A1_P2} in \eqref{eq:monotonicity_A1}, gathering terms and recalling that $\gamma_\lambda = \frac{\lambda}{C_\lambda}$ and $\mu_\lambda=\frac{1}{C_\lambda^2}$. It follows

\begin{equation}
\begin{aligned}\label{eq:A_1-dissip}
     \left\langle A_1 z,z\right\rangle_{\mu}  &\geq \frac{1}{2}\|\nabla z\|^2_{L^2(\Omega)}+\left(  \frac{\lambda}{C_\lambda^2} - \frac{\lambda\sqrt{2}}{2C^2_\lambda}\right) \| P_N z \|^2_{L^{2}(\Omega)} + \left( \frac{\lambda}{2}- \frac{{\lambda\sqrt{2}}}{4}  \right)\| P_N^\perp z\|^2_{L^{2}(\Omega)}\\[2ex]
     &\geq \frac{1}{2}\|\nabla z\|^2_{L^2(\Omega)}+\frac{\lambda}{C_\lambda^2}\left(  1 - \frac{\sqrt{2}}{2}\right) \| P_N z \|^2_{L^{2}(\Omega)} + \frac{\lambda}{2}\left( 1- \frac{\sqrt{2}}{2  }  \right)\| P_N^\perp z \|^2_{L^{2}(\Omega)} \\[2ex]
     &= C\left(\|\nabla z\|^2_{L^2(\Omega)}+ \|  z \|^2_{L^{2}(\Omega)}\right)  =C\|z\|^2_{H_0^1(\Omega)}\geq 0,
\end{aligned}    
\end{equation}
where $C=\min\left\{\frac{1}{2},\frac{\lambda}{C_\lambda^2}\left(  1 - \frac{\sqrt{2}}{2}\right),\frac{\lambda}{2}\left( 1- \frac{\sqrt{2}}{2  }  \right)\right\}$.

\noindent For the nonlinear part, we have

\begin{align*}
    \langle B y_1- By_2, y_1-y_2\rangle_\mu& = \mu_\lambda \langle\chi_\omega\eta_1-\chi_\omega\eta_2,P_N (y_1-y_2)\rangle_{L^{2}(\Omega)}+ \langle\chi_\omega\eta_1-\chi_\omega\eta_2, P_N^\perp (y_1-y_2)\rangle_{L^{2}(\Omega)} \\[2ex]
    & =  \langle\eta_1-\eta_2, \chi_\omega\mu P_N (y_1-y_2)\rangle_{L^{2}(\Omega)}+ \langle\eta_1-\eta_2, \chi_\omega P_N^\perp (y_1-y_2)\rangle_{L^{2}(\Omega)} \\[2ex]
    &= \langle\eta_1-\eta_2,\chi_\omega \mu_\lambda P_N (y_1-y_2)+\chi_\omega (P_N^\perp (y_1-y_2))\rangle_{L^{2}(\Omega)} \\[2ex]
    &=\langle\eta_1-\eta_2,\chi_\omega\mathscr{C}y_1 -\chi_\omega \mathscr{C} y_2  \rangle_{L^{2}(\Omega)}\geq 0.
\end{align*}

\noindent The last inequality is a consequence of the multivalued sign operator. 

\noindent Therefore, we get

\begin{equation}\label{eq:A-dissip}
    \left\langle Ay_1 -A y_2,y_1-y_2\right\rangle_{\mu}  \geq \langle A_1 y_1 -A_1 y_2 ,y_1-y_2\rangle_\mu+ \langle B y_1 +By_2, y_1-y_2\rangle_\mu\geq 0.  
\end{equation}
\end{proof}

\begin{rem}
    The choice of the $\mu$-inner product is particularly useful in the nonlinear part of the proof. This definition yields a natural monotonicity property, which is key to handling the sign operator.
\end{rem}


\begin{prop}\label{wp2}
The operator $A$ defined by \ref{operator} satisfies $ R(I+A) = L^2(\Omega) $.
\end{prop}


\begin{proof}

Given a $ f \in L^2(\Omega) $, we need to show the existence of $ y \in D(A) $ such that
\begin{equation*}
y + Ay \ni f,
\end{equation*}

\noindent or equivalently

\begin{equation}\label{eq:main_inclusion}
y - \Delta y  +\gamma_\lambda \chi_\omega P_N(y)+\chi_\omega D \text{sign}_{L^2(\Omega)}(\chi_\omega \mathscr{C}y) \ni f.
\end{equation}

\noindent To handle the set-valued signum nonlinearity, we employ a regularization argument based on the Yosida approximation. Let $ \varphi(y) = D \|y\|_{L^2(\Omega)} $, whose subdifferential is $ \partial \varphi(y) = D \cdot \operatorname{sign}_{L^2(\Omega)}(y) $. For $ \sigma > 0 $, the Moreau Regularization  \cite[Chapter IV, Proposition 1.8]{s1997} of $\varphi(y) $ is given by

\begin{align*}
\varphi_\sigma(y) &= \min \left\{ \frac{1}{2\sigma} \|y - z\|^2 + \varphi(z) : z \in L^2(\Omega) \right\} \\[2ex]
&= \frac{1}{\sigma} \|y-J_\sigma y\|^2 + \varphi(J_\sigma(y)) \\[2ex]
&= \frac{\sigma}{2} \|\alpha_\sigma(y)\|^2 + \varphi(J_\sigma(y)),
\end{align*}

\noindent where $\alpha_\sigma $ is the Yosida approximation and $J_\sigma$ is the resolvent. Moreover, we have $ \alpha_\sigma(y) = (\partial \varphi)_\sigma(y)=\nabla\varphi_\sigma(y)$. Now, we define the regularized operator

\begin{equation*}
 B_\sigma(y) := D \chi_\omega \alpha_\sigma(\chi_\omega \mathscr{C} y).   
\end{equation*}

\noindent Note that $ B_\sigma : L^2(\Omega) \to L^2(\Omega) $ is single-valued and satisfies $ \|B_\sigma(y)\|_{L^2(\Omega)} \le D $ for all $ y \in L^2(\Omega) $. Indeed, by the properties of the Yosida approximation,  we have $\|\alpha_\sigma(y)\|_{L^2(\Omega)}=\|\partial \varphi_\sigma(y)\|_{L^2(\Omega)}$ and $\|\partial \varphi_\sigma(y)\|_{L^2(\Omega)}\leq\|\partial \varphi(y)\|_{L^2(\Omega)}$. Then, using the definition of the multivalued sign operator, for any $\eta \in \operatorname{sign}_{L^2(\Omega)}\left(\chi_\omega \mathscr{C} u\right)$, we have to consider two cases

\begin{enumerate}
    \item If \(\chi_\omega \mathscr{C} u \neq 0\) in \(L^2(\Omega)\), then
    
   \begin{equation*}
    \eta = \frac{\chi_\omega \mathscr{C} u}{\|\chi_\omega \mathscr{C} u\|_{L^2(\Omega)}}.   
   \end{equation*}
   
   By construction, \(\|\eta\|_{L^2(\Omega)} = 1\).

    \item If \(\chi_\omega \mathscr{C} u = 0\) in \(L^2(\Omega)\), then:
    
   \begin{equation*}
       \eta \in \left\{ h \in L^2(\Omega) \mid \|h\|_{L^2(\Omega)} \leq 1 \right\}.
   \end{equation*}
   
   Thus, $\|\eta\|_{L^2(\Omega)} \leq 1$.
\end{enumerate} 

\noindent Therefore,  
   \begin{equation*}
      \|D \chi_\omega \eta\|_{L^2(\Omega)} = D \left( \int_\omega |\eta|^2 \, dx \right)^{1/2}\leq D \left(\int_\Omega |\eta|^2 \, dx\right)^{1/2} \leq D.  
   \end{equation*}

\noindent We now consider the regularized equation: find $ y_\sigma \in H_0^1(\Omega) $ such that

\begin{equation} \label{eq:regularized}
y_\sigma - \Delta y_\sigma + \gamma_\lambda \chi_\omega P_N(y_\sigma) = f - B_\sigma(y_\sigma).
\end{equation}


\noindent We proceed via the Schauder Fixed Point Theorem. Consider the map

\begin{align*}
T_\sigma \colon L^2(\Omega) &\longrightarrow L^2(\Omega) \\
 u & \longmapsto  T(u)=y_{\sigma, u}~,
\end{align*}

\noindent where $y_{\sigma, u}$ is solution of

\begin{equation} \label{eq:fixed_pt}
y_{\sigma, u} - \Delta y_{\sigma, u} + \gamma_\lambda \chi_\omega P_N(y_{\sigma, u}) = f - B_\sigma(u).
\end{equation}

\noindent For a regular solution $y$ of \eqref{eq:fixed_pt} we take the $\mu-$inner product with a test function $z$, to derive the variational formulation $a(y,z)=(f-B_\sigma,z)_\mu$,  where the bilinear form $a: ~H_0^1(\Omega)\times H_0^1(\Omega)\rightarrow \mathbb{R}$ is given by 

\begin{equation*}
 a(y, z)=(y,z)_\mu+(\nabla y, \nabla z)_\mu+ (\gamma_\lambda \chi_\omega P_N(y),z)_\mu,
\end{equation*}

\noindent We have that $ a(\cdot, \cdot) $ is continuous and coercive on $ H_0^1(\Omega) $. Indeed, for $y\in H_0^1(\Omega),$ we have

\begin{equation*}
 a(y, y)=\|y\|^2_\mu+\|\nabla y\|^2_\mu+ \gamma_\lambda \mu \| P_N y \|^2_{L^2(\omega)} +\gamma_\lambda\int_\omega P_N (y) P_N^\perp (y) ~dx.
\end{equation*}

\noindent By using \eqref{eq:A_1-dissip}, we have

\begin{equation*}
    a(y, y)\geq\|y\|^2_\mu +C\|y\|^2_{H_0^1(\Omega)}\geq \beta \|y\|^2_{H_0^1(\Omega)},
\end{equation*}
where $\beta=\min\lbrace 1, C \rbrace$.

\noindent For any $ u \in L^2(\Omega) $, the right-hand side of \eqref{eq:fixed_pt}, $ f - B_\sigma(u) $, belongs to $ L^2(\Omega) $. Thus, by the Lax-Milgram Theorem, there exists a unique solution $ y_{\sigma, u} \in H_0^1(\Omega)$ to the variational formulation

\begin{equation*}
 (y,z)_\mu+(\nabla y, \nabla z)_\mu+(\gamma_\lambda \chi_\omega P_N(y),z)_\mu=(f-B_\sigma(u),z)_\mu, \qquad \forall z \in H_0^1(\Omega).
\end{equation*}

\noindent By definition on $\mu -$inner product, we can note that

\begin{equation*}
 (w,v)_\mu=\mu_\lambda(w, P_N v)_{L^2(\Omega)}+(w, P_N^\perp v)_{L^2(\Omega)}=(w,\mu_\lambda P_N v +P_N^\perp v)_{L^2(\Omega)}=(w,\mathscr{C}v)_{L^2(\Omega)}.   
\end{equation*}

\noindent  Now, since $\nabla$ and $\mathscr{C}$ commute, we have

\begin{equation*}
 (y,\mathscr{C}z)_{L^2(\Omega)}+(\nabla y, \nabla \mathscr{C}z)_{L^2(\Omega)}+(\gamma_\lambda \chi_\omega P_N(y),\mathscr{C}z)_{L^2(\Omega)}=(f-B_\sigma(u),\mathscr{C}z)_{L^2(\Omega)}, \qquad \forall z \in H_0^1(\Omega).
\end{equation*}

\noindent Notice that $\mathscr{C}:H_0^1(\Omega)\rightarrow H_0^1(\Omega)$ is an isomorphism (Appendix \ref{iso_C}). Then the variational formulation

\begin{equation*}
 a(y, \mathscr{C} z)=F(\mathscr{C} z), \quad \forall z \in H_0^1   
\end{equation*}

\noindent is equivalent to

\begin{equation*}
  a(y, w)=F(w), \quad \forall w \in H_0^1,  
\end{equation*}

\noindent by taking $w=\mathscr{C} z$ and noting $z=\mathscr{C}^{-1} w \in H_0^1$. It follows that

\begin{equation*}
    (y,w)_{L^2(\Omega)}+(\nabla y, \nabla w )_{L^2(\Omega)}+ \gamma_\lambda(\chi_\omega P_N(y), w)_{L^2(\Omega)}=(f- B_\sigma (u), w)_{L^2(\Omega)},\qquad \forall w\in H_0^1(\Omega).
\end{equation*}

\noindent From which we have $y_{\sigma, u}-\Delta y_{\sigma, u}+\gamma_\lambda \chi_\omega P_N (y_{\sigma, u})=f-B_\sigma(u),$ in the sense of distributions.  

\noindent Since $y_{\sigma, u} \in H_0^1(\Omega)$, the term $\gamma_\lambda \chi_\omega P_N (y_{\sigma, u})$ belongs to $L^2(\Omega)$, as $P_N (y_{\sigma, u}) \in L^2(\Omega)$ and $\chi_\omega$ is bounded. Therefore, the entire right-hand side $f-B_\sigma(u)-\gamma_\lambda \chi_\omega P_N (y_{\sigma, u})$ lies in $L^2(\Omega)$. By standard elliptic regularity theory~\cite[Theorem 9.25]{b2010}, we conclude that $y_{\sigma, u} \in H^2(\Omega)$. Moreover, since $y_{\sigma, u} \in H^2(\Omega)$, we get

\begin{equation}\label{eq:ec_punt}
 -\Delta y_{\sigma, u}+y_{\sigma, u}+\gamma_\lambda \chi_\omega P_N y_{\sigma, u}=f-B_\sigma(u) \text{ a.e. on }\Omega.   
\end{equation}

\noindent Now, Let  $M > 0$, to be selected later, and define $K_M := \{ v \in H_0^1(\Omega)\cap H^2(\Omega): \|v\|_{H^2(\Omega)} \le M \}$. 

\begin{itemize}
    \item Upper bound for $\|y_{\sigma,u}\|_{L^2(\Omega)}$
independent of $\sigma\in(0,\infty)$.
\end{itemize}

\noindent Let $y_{\sigma, u}\in K_M$, testing \eqref{eq:fixed_pt} with $ y_{\sigma, u} $ in the $\mu-$ inner product gives

\begin{equation}\label{eq:bound_H1}
(y_{\sigma,u} , y_{\sigma, u})_\mu+(-\Delta y_{\sigma, u},y_{\sigma, u})_\mu+\gamma_\lambda(\chi_\omega P_N y_{\sigma, u},y_{\sigma, u})_\mu=\left( f-B_\sigma(u), y_{\sigma, u}\right)_{\mu} .    
\end{equation}

\noindent By using Cauchy-Schwarz and Holder inequalities in the right-hand side of \eqref{eq:bound_H1}, we obtain the a priori estimate
\begin{equation}\label{eq:bound_right_L2}
 \begin{aligned}
\left( f-B_\sigma(u), y_{\sigma, u}\right)_{\mu} &\leq \mu\left( f-B_\sigma(u), P_N y_{\sigma, u}\right)_{L^2(\Omega)}+\left( f-B_\sigma(u), P_N^\perp y_{\sigma, u}\right)_{L^2(\Omega)}\\[2ex]
&\leq \mu\left\| f-B_\sigma(u)\|_{L^2(\Omega)}  \| P_N y_{\sigma, u}\right\|_{L^2(\Omega)}+\left\| f-B_\sigma(u)\right\|_{L^2(\Omega)} \left\| P_N^\perp y_{\sigma, u}\right\|_{L^2(\Omega)}\\[2ex]
&\leq \left(\frac{\mu}{2}+\frac{1}{2}\right)\|f-B_\sigma (u)\|^2_{L^2(\Omega)}+\frac{\mu_\lambda}{2}\|P_N y_{\sigma, u}\|^2_{L^2(\Omega)}+\frac{1}{2}\|P_N^\perp y_{\sigma, u}\|^2_{L^2(\Omega)}\\[2ex]
&\leq \frac{1}{2}(\mu_\lambda +1)\left(2\|f\|^2_{L^2(\Omega)}+2\|B_\sigma (u)\|^2_{L^2(\Omega)}\right)+\frac{\mu_\lambda}{2}\|P_N y_{\sigma, u}\|^2_{L^2(\Omega)}+\frac{1}{2}\|P_N^\perp y_{\sigma, u}\|^2_{L^2(\Omega)}.
\end{aligned}   
\end{equation}

\noindent On the other hand, note that $(y_{\sigma,u} , y_{\sigma, u})_\mu=\mu_\lambda\|P_N y_{\sigma, u}\|_{L^2(\Omega)}^2+\|P_N^\perp y_{\sigma, u}\|_{L^2(\Omega)}^2$, and by using \eqref{eq:A_1-dissip} in \eqref{eq:bound_H1} we get $(-\Delta y_{\sigma, u},y_{\sigma, u})_\mu+\gamma_\lambda(\chi_\omega P_N y_{\sigma, u},y_{\sigma, u})_\mu\geq0$. Then, it follows

\begin{equation*}
\begin{aligned}
  &\mu_\lambda\|P_N y_{\sigma, u}\|_{L^2(\Omega)}^2+\|P_N^\perp y_{\sigma, u}\|_{L^2(\Omega)}^2\\&\leq   (\mu_\lambda +1)\left(\|f\|^2_{L^2(\Omega)}+\|B_\sigma (u)\|^2_{L^2(\Omega)}\right)+\frac{\mu_\lambda}{2}\|P_N y_{\sigma, u}|^2_{L^2(\Omega)}+\frac{1}{2}\|P_N^\perp y_{\sigma, u}\|^2_{L^2(\Omega)}.
  \end{aligned}
\end{equation*}

\noindent Then

\begin{equation*}
  \frac{\mu_\lambda}{2}\|P_N y_{\sigma, u}\|^2_{L^2(\Omega)}+\frac{1}{2}\|P_N^\perp y_{\sigma, u}\|^2_{L^2(\Omega)}\leq   (\mu_\lambda +1)\left(\|f\|^2_{L^2(\Omega)}+\|B_\sigma (u)\|^2_{L^2(\Omega)}\right).
\end{equation*}

\noindent Since $0<C<1$ and $\mu_\lambda=\frac{1}{C^2}>1 $, then $\mu_\lambda>1$, and by using the boundedness of $B_\sigma$, we get

\begin{equation}\label{eq:bound_L2}
    \| y_{\sigma, u}\|_{L^2}^2 \leq  4\mu_\lambda \left(\|f\|_{L^2(\Omega)}^2+D^2\right) .
\end{equation}

\begin{itemize}
    \item Upper bound for $\|y_{\sigma,u}\|_{H^1(\Omega)}$
independent of $\sigma\in(0,\infty)$.
\end{itemize}

\noindent Testing \eqref{eq:ec_punt} with $ y_{\sigma, u} $ in the $L^2$- inner product gives

\begin{equation}\label{eq:bound_H1-L2}
(y_{\sigma,u} , y_{\sigma, u})+(-\Delta y_{\sigma, u},y_{\sigma, u})=\left( f-B_\sigma(u), y_{\sigma, u}\right) -\gamma_\lambda(\chi_\omega P_N y_{\sigma, u},y_{\sigma, u}).    
\end{equation}

\noindent On the left side, by integrating by parts, we have

\begin{align*}
    (y_{\sigma,u} , y_{\sigma, u})+(-\Delta y_{\sigma, u},y_{\sigma, u})=\|y_{\sigma,u}\|^2_{L^2(\Omega)}+\|\nabla y_{\sigma,u}\|^2_{L^2(\Omega)}=\|y\|_{H^1(\Omega)}^2.
\end{align*}

\noindent On the right side, we have 
\begin{align*}
    \left( f-B_\sigma(u), y_{\sigma, u}\right) -\gamma_\lambda(\chi_\omega P_N y_{\sigma, u},y_{\sigma, u})&\leq \|f-B_\sigma(u)\|_{L^2(\Omega)}\|y_{\sigma, u}\|_{L^2(\Omega)}+\gamma_\lambda\|P_N y_{\sigma, u}\|_{L^2(\omega)}\|y_{\sigma, u}\|_{L^2(\omega)}\\[2ex]
    &\leq \frac{1}{2}\|f-B_\sigma(u)\|^2_{L^2(\Omega)}+\frac{1}{2}\|y_{\sigma, u}\|^2_{L^2(\Omega)}+\gamma_\lambda\|y_{\sigma, u}\|^2_{L^2(\Omega)}\\[2ex]
    &\leq \|f\|^2_{L^2(\Omega)}+\|B_\sigma(u)\|^2_{L^2(\Omega)}+\left(\frac{1}{2}+\gamma_\lambda\right)\|y_{\sigma, u}\|^2_{L^2(\Omega)}.
\end{align*}

\noindent Therefore, using the previous inequalities, the bound of $B_\sigma(u)$ and \eqref{eq:bound_L2}, we have

\begin{align*}
\|y_{\sigma,u}\|_{H^1(\Omega)}^2&\leq\|f\|^2_{L^2(\Omega)}+D^2+4\mu_\lambda \left(\frac{1}{2}+\gamma_\lambda\right)  \left(\|f\|_{L^2(\Omega)}^2+D^2\right)\\[2ex]
&=\left(1+4\mu_\lambda\left(\frac{1}{2}+\gamma_\lambda\right) \right)\left(\|f\|_{L^2(\Omega)}^2+D^2\right).
\end{align*}

\begin{itemize}
    \item Upper bound for $\|y_{\sigma,u}\|_{H^2(\Omega)}$
independent of $\sigma\in(0,\infty)$.
\end{itemize}

\noindent By using \eqref{eq:ec_punt} and integrating on $\Omega$, we get

\begin{equation*}
\begin{aligned}
    \|\Delta y_{\sigma,u}\|^2_{L^2(\Omega)}&=\left\|y_{\sigma,u}+\gamma_\lambda\chi_\omega P_N y_{\sigma,u} +B_\sigma (y_{\sigma,u})-f\right\|_{L^2(\Omega)}^2 \\[2ex]
    &\leq 4\left(\|y_{\sigma,u}\|_{L^2(\Omega)}^2+\gamma_\lambda^2\|\chi_\omega P_N y_{\sigma,u}\|_{L^2(\Omega)}^2+\|f\|_{L^2(\Omega)}^2 +\|B_\sigma (y_{\sigma,u})\|_{L^2(\Omega)}^2\right)\\[2ex]
    &\leq 4\left((1+\gamma_\lambda^2)\|y_{\sigma,u}\|_{L^2(\Omega)}^2+\|f\|_{L^2(\Omega)}^2 +\|B_\sigma (y_\sigma)\|_{L^2(\Omega)}^2\right).
    \end{aligned}
\end{equation*}

\noindent Therefore, by using the bound of $B_\sigma$ and $y_\sigma$, we get

\begin{equation*}
\begin{aligned}
    \|\Delta y_\sigma\|^2_{L^2(\Omega)} &\leq 16(1+\gamma_\lambda^2)\mu_\lambda\left(\|f\|_{L^2(\Omega)}^2+D^2\right)+4\left(\|f\|_{L^2(\Omega)}^2 +D^2\right)\\[2ex]
    &=4(4(1+\gamma_\lambda^2)\mu_\lambda+1)\left(\|f\|_{L^2(\Omega)}^2+D^2\right).
    \end{aligned}
\end{equation*}

\noindent Choosing $M\geq  4(4(1+\gamma_\lambda^2)\mu_\lambda+1) \left(\|f\|_{L^2(\Omega)}^2+D^2\right)$, we have $T_{\sigma}(K_M)\subset K_M.$ Since $\{y_{\sigma,u}\}\subset K_M$ is bounded in $H^2(\Omega)$, by the Banach–Alaoglu theorem, there exists a subsequence, still denoted $y_{\sigma,u}$, converging weakly in $H^2(\Omega)$ to some $w \in H^2(\Omega)$, i.e., $y_{\sigma,u}\rightharpoonup w$ in $H^2(\Omega)$. By Rellich-Kondrachov theorem~\cite[Theorem 9.16, page 285]{HB2010}, we know that the embedding $H^2(\Omega) \hookrightarrow L^2(\Omega)$ is compact, so $y_{\sigma,u} \to w$ strongly in $L^2(\Omega)$. Moreover, since $y_{\sigma,u} \in H^1_0(\Omega)$ and $H^1_0(\Omega)$ is a closed subspace of $H^1(\Omega)$, and strong $L^2$ convergence with boundedness in $H^1$ implies weak $H^1$ convergence to the same limit, we get $w \in H^1_0(\Omega)$. Therefore, $w \in K_M$, so $K_M$ is closed in $L^2(\Omega)$. Thus, we conclude that $K_M$ is compact $L^2(\Omega)$.

\noindent By the Schauder Fixed Point Theorem , $T_\sigma$ has a fixed point $y_\sigma\in K_M\subset H^2(\Omega)\cap H_0^1(\Omega),$ which is a solution to \eqref{eq:fixed_pt}. 

\vspace{0.5cm}

Accordingly, we have shown


\begin{lem}\label{lemma_regularized}
For any $\sigma \in(0, \infty)$ there exists $y_\sigma \in H^2(\Omega)\cap H_0^1(\Omega)$ such that $y_\sigma-\Delta y_\sigma=f -B\sigma (y_\sigma)-\gamma_\lambda\chi_\omega P_N y_\sigma$ for almost every $x \in\Omega$.
\end{lem}


\noindent Let us consider the $y_\sigma$ given by Lemma \ref{lemma_regularized}. We proceed to prove
that $R(I+A) = L^2(\Omega) $ by analyzing what happens to $y_\sigma$ as $\sigma\rightarrow 0^+.$ By using a priori estimates as before, we obtain that the sequence $ \{y_\sigma\} $ is uniformly bounded in $ H_0^1(\Omega) $, i.e.,

\begin{align*}
\|y_{\sigma}\|_{H^1(\Omega)}^2&\leq M.
\end{align*}

\noindent Therefore, there exists a subsequence (still denoted by $ \{y_\sigma\} $) and a function $ y \in H_0^1(\Omega) $ such that

\begin{align*}
y_\sigma &\rightharpoonup y \quad \text{in } H_0^1(\Omega), \\
y_\sigma &\to y \quad \text{in } L^2(\Omega).
\end{align*}

\noindent By the continuity of the operators $ \mathscr{C} $ and $ \chi_\omega $, it follows that $ \chi_\omega \mathscr{C} y_\sigma \to \chi_\omega \mathscr{C} y $ in $ L^2(\Omega) $.

\noindent From the bound $ \|B_\sigma(y_\sigma)\|_{L^2(\Omega)} \le D $, there exists a $ g \in L^2(\Omega) $ and a further subsequence such that

\begin{equation*}
B_\sigma(y_\sigma) \rightharpoonup g \quad \text{in } L^2(\Omega).    
\end{equation*}

\noindent We can now pass to the limit in the weak formulation of \eqref{eq:regularized}. For any test function $ z \in H_0^1(\Omega) $, we have

\begin{equation*}
\int_\Omega y_\sigma z + \int_\Omega\nabla y_\sigma  \nabla z \, dx + \gamma_\lambda \int_\omega P_N(y_\sigma) z \, dx = \int_\Omega (f - B_\sigma(y_\sigma)) z \, dx.
\end{equation*}

\noindent Taking the limit $ \sigma \to 0^+ $, we obtain

\begin{equation*}
\int_\Omega y z + \nabla y \cdot \nabla z \, dx + \gamma_\lambda \int_\omega P_N(y) z \, dx = \int_\Omega (f - g) z \, dx.
\end{equation*}

\noindent This implies the equation holds in the sense of distributions,  we have

\begin{equation} \label{eq:limit_equation}
y - \Delta y + \gamma_\lambda \chi_\omega P_N y = f - g .
\end{equation}

 Accordingly, Lemma \ref{lemma_regularized} and the previous arguments yield


\begin{lem}\label{lemma:problem_with_g}
 There exists $y \in H^2(\Omega)\cap H_0^1(\Omega)$ such that $y-\Delta y+\gamma_\lambda \chi_\omega P_N y = f - g,$ for almost every $x \in\Omega$.   
\end{lem}


\noindent In view of Lemma \ref{lemma:problem_with_g}, we see that in order to complete the proof of $R(I+\mathcal{A})=L^2(0, L)$ it remains to show that $ g \in D \chi_\omega  \operatorname{sign}_{L^2(\Omega)}(\chi_\omega \mathscr{C} y) $. Let $ \xi_\sigma = \alpha_\sigma(\chi_\omega \mathscr{C} y_\sigma) $, so that $ B_\sigma(y_\sigma) = D \chi_\omega \xi_\sigma $. By the properties of the Yosida approximation, $ \xi_\sigma \in \partial \varphi(J_\sigma(\chi_\omega \mathscr{C} y_\sigma)) = \operatorname{sign}_{L^2(\Omega)}(J_\sigma(\chi_\omega \mathscr{C} y_\sigma)) $. Consequently, $ \|\xi_\sigma\|_{L^2(\Omega)} \le 1 $. Thus, on a subsequence,

\begin{equation*}
 \xi_\sigma \rightharpoonup \xi \quad \text{in } L^2(\Omega), \quad \text{with } \|\xi\|_{L^2(\Omega)} \le 1.   
\end{equation*}

\noindent Furthermore, we have the estimate

\begin{equation*}
\| \chi_\omega \mathscr{C} y_\sigma - J_\sigma(\chi_\omega \mathscr{C} y_\sigma) \|_{L^2(\Omega)} = \sigma \|\alpha_\sigma(\chi_\omega \mathscr{C} y_\sigma)\|_{L^2(\Omega)} = \sigma \|\xi_\sigma\|_{L^2(\Omega)} \le \sigma \to 0,    
\end{equation*}

\noindent and since $ \chi_\omega \mathscr{C} y_\sigma \to \chi_\omega \mathscr{C} y $, it follows that $ J_\sigma(\chi_\omega \mathscr{C} y_\sigma) \to \chi_\omega \mathscr{C} y $ in $ L^2(\Omega) $. Indeed,

\begin{equation}
\begin{aligned}
 \left\|J_\sigma(\chi_\omega \mathscr{C} y_\sigma)-\chi_\omega \mathscr{C} y\right\|_{L^2(\Omega)} & =\left\|J_\sigma(\chi_\omega \mathscr{C} y_\sigma)-\chi_\omega \mathscr{C} y_\sigma+\chi_\omega \mathscr{C} y_\sigma-\chi_\omega \mathscr{C} y\right\|_{L^2(\Omega)} \\[2ex]
& \leq\left\|J_\sigma(\chi_\omega \mathscr{C} y_\sigma)-\chi_\omega \mathscr{C} y_\sigma\right\|_{L^2(\Omega)}+\left\|\chi_\omega \mathscr{C} y_\sigma-\chi_\omega \mathscr{C} y\right\|_{L^2(\Omega)}\\[2ex]
&\rightarrow 0.
\end{aligned}
\end{equation}

\noindent Now, by using \cite[Proposition 1.6.]{s1997} we have that if $ \xi_\sigma \in \operatorname{sign}(J_\sigma(\chi_\omega \mathscr{C} y_\sigma)) $, $ \xi_\sigma \rightharpoonup \xi $, and $ J_\sigma(\chi_\omega \mathscr{C} y_\sigma) \to \chi_\omega \mathscr{C} y $, it follows that $\xi \in \operatorname{sign}_{L^2(\Omega)}(\chi_\omega \mathscr{C} y)  $. Finally, since $B_\sigma\left(y_\sigma\right)=D \chi_\omega \xi_\sigma \rightarrow D \chi_\omega \xi$ and we have $B_\sigma\left(y_\sigma\right) \rightharpoonup g$, the uniqueness of the weak limit implies $g=D \chi_\omega \xi$. Substituting into \eqref{eq:limit_equation}, we conclude that $ y $ satisfies \eqref{eq:main_inclusion}, completing the proof. 

\end{proof}


\begin{rem}
An interesting approach to proving that the operator \eqref{operator} is maximal monotone is presented in \cite[Lemma 3.]{lr2025}, which consist in proving that $B \operatorname{sign} (B)$ is maximal monotone, where the lineal closed operator $B$ considered in \cite{lr2025}, has to satisfies $ B^2 = B$ and $ B = B^*$. However, this method fails in our setting, since $\chi_\omega \mathscr{C}$ under consideration in our work fails to satisfy the requisite conditions.
\end{rem}


\noindent In  order to complete the proof of Theorem \ref{main_theorem} we proceed to show that \eqref{P} with the feedback law \eqref{feedback} has a unique weak solution $y \in C\left([0, \infty) ; L^2(\Omega)\right)$ and it satisfies \eqref{eq:decay_solution}.

\noindent Let us return to \eqref{system}. In view of Proposition \ref{wp1} and Proposition \ref{wp2}, we can apply \cite[Chapter IV, Lemma 1.3]{s1997} to conclude that the operator $\mathcal{A}$ defined in \eqref{operator} is maximal monotone. For the moment, let us assume that $d \in W^{1,1}\left(0, \infty ; L^2(\Omega)\right)$ and $y_0 \in D(A)$. Then,~\cite[Chapter IV, Theorem 4.1]{s1997} gives the existence of a unique $y \in W^{1,1}\left(0, \infty ; L^2(\Omega)\right)$ such that

\begin{equation}\label{eq:resulting_system}
\begin{cases}
 y^{\prime}(t)+A y(t) \ni \chi_\omega d(t) \text { for almost every } t>0 \\
 y(t) \in D(A) \text { for every } t \geq 0\\
 y(0)=y_0 .
\end{cases}    
\end{equation}


\noindent Given $(d, \widehat{d}) \in W^{1,1}\left(0, \infty ; L^2(\Omega)\right)^2$ and $\left(y_0, \widehat{y}_0\right) \in D(A)^2$, let us take the corresponding unique solution $(y, \widehat{y}) \in W^{1,1}\left(0, \infty ; L^2(\Omega)\right)^2$. Then, by \cite[Chapter IV, (4.12)]{s1997}, we have

\begin{equation}\label{eq:Showalter_4.12}
\|y(t, \cdot)-\widehat{y}(t, \cdot)\|_{L^2(\Omega)} \leq\left\|y_0-\widehat{y}_0\right\|_{L^2(\Omega)}  +\int_0^t\|d(s, \cdot)-\widehat{d}(s, \cdot)\|_{L^2(\omega)}~ d s, ~t \geq 0.
\end{equation}

\noindent We may use \eqref{eq:Showalter_4.12} to define, as in \cite[Page 183]{s1997} for instance, the notion of weak solution of \eqref{system}.
\begin{defi} A generalized solution of \eqref{system} is a function $y \in C([0, T], L^2(\Omega))$ for which there exists a sequence of (absolutely continuous) solutions $y_n$ of

\begin{equation*}
\frac{d y_n}{d t}+A\left(y_n\right) \ni \chi_\omega d_n,\quad n \geq 1    
\end{equation*}
with $d_n \rightarrow d$ in $L^1(0, T ; L^2(\Omega))$ and $y_n \rightarrow y$ in $C([0, T], L^2(\Omega))$.
\end{defi}

Taking into account the density of $W^{1,1}\left(0, \infty ; L^2(\Omega)\right)$ in $L^1\left(0, \infty ; L^2(\Omega)\right)$ and of $D(A)$ in $L^2(\Omega)$, we have that generalized solution is well defined.

Therefore, in virtue of \cite[Chapter IV, Theorem 4.1A]{s1997} we have that \eqref{eq:closed_loop_system} has a unique weak solution $y \in C\left([0, \infty) ; L^2(\Omega)\right)$ provided that $d \in L^1\left(0, \infty ; L^2(\Omega)\right)$ and $y_0 \in L^2(\Omega)$. Moreover, all the formal computations done in Section \ref{feedback_design} make sense, implying that \eqref{eq:decay_solution} is satisfied. 

Accordingly, we have shown
\begin{prop}
Let $d \in L^1(0, \infty ; L^2(\Omega))$ and $y_0 \in L^2(\Omega)$. Then, there exists a unique $y=y(t, x)$ such that:
\begin{enumerate}
    \item $y \in C([0, \infty) ; L^2(\Omega))$.
    \item  $y(0)=y_0$.
    \item  Any two generalized solutions of the problem \eqref{system} satisfies the estimate \eqref{eq:Showalter_4.12}.
    \item  $y(t, x)$ satisfies the exponential decay \eqref{eq:decay_solution} of Theorem \ref{main_theorem}.
\end{enumerate}
\end{prop}

\noindent Concluding the proof of the well-posedness part of Theorem \ref{main_theorem}.



\section{Concluding Remarks and Perspectives}
\label{c}

In this work, we presented a rapid stabilization strategy for the heat equation under unknown disturbances with a localized internal feedback law. We have employed the method introduced by \cite{X2024} combined with the sign multivalued operator \eqref{sign} to design suitable feedback laws that guarantee exponential decay with an arbitrary decay rate. The assumptions made on the unknown disturbance (Assumptions \ref{A_1} and \ref{A_2}) are the standard ones that can be found in the literature. The corresponding closed-loop system \eqref{eq:closed_loop_system} is formulated as a differential inclusion, and its well-posedness is proved via maximal monotone operator theory. The main difficulty in the application of the maximal monotone operator theory comes from the nonlinearity of the sign operator. To deal with the monotonicity, a new inner product related to the Frequency Lyapunov Functional was used; the maximality was treated by regularizing the operator via the Yosida Approximation and then applying a fixed-point argument.

\noindent Following the steps of the previous section, we could prove the exponential stability in the case where a potential is included that makes the system unstable, i.e., the equation

\begin{equation}\label{Pot}
  \left\{\begin{array}{ll}
y_t- \Delta y+ a(x)y=\chi_\omega(u+d),\quad&(t, x) \in(0, \infty) \times \Omega \\
y(t, x)=0,&(t, x) \in(0, \infty) \times \partial \Omega \\
y(0, x)=y_0(x), &x \in \Omega
\end{array}\right.
\end{equation}
where $c\in L^\infty(\Omega)$ is an extra source of instability. In that case, we get a section of the eigenvalues to be non-positive. The calculations follow in a similar way to the case without a potential, the election of the same feedback law and Lyapunov functional works, considering the same total control as before, given by the decomposition 
$u(t, x)  = \tilde{u}+\hat{u}$. The crucial step in this case is the election of the constants $\mu_\lambda $ and $\gamma_\lambda$, where the proposed election is 

\begin{equation}
\gamma_\lambda := \frac{1}{C}(\lambda - \tau_1), \qquad \mu_\lambda := \frac{\gamma_\lambda^2}{\lambda^2} = \frac{(\lambda - \tau_1)^2}{\lambda^2 C^2},     
\end{equation}

\noindent where $\tau_1$ is the first of the non-positive eigenvalues, i.e., $ -\tau_1 \geq \dots \geq-\tau_j\geq 0$.

\noindent Another interesting approach is to study the problem in divergence form, i.e.,

\begin{equation*}
\left\{\begin{array}{ll}
y_t - div\big(\gamma(x)\cdot\nabla y\big) = \chi_\omega(u+d), & (t, x) \in (0, \infty) \times \Omega, \\[4pt]
y(t, x) = 0, & (t, x) \in (0, \infty) \times \partial \Omega, \\[4pt]
y(0, x) = y_0(x), & x \in \Omega,
\end{array}\right.
\end{equation*}
where $\gamma \in C^2(\bar{\Omega})$ with $\gamma(x) \geq \gamma_0>0$ in $\Omega$, for which spectral inequalities remain fundamentally important. In particular, a case of considerable interest is based on the spectral inequality introduced by Osses and Triki \cite{ot2025}, which refines the classical spectral inequality of \cite{lr1995}, also employed in \cite{X2024}. Their result not only recovers the classical base case, but the bound in terms of the frequency number is more precise since it depends on the coefficients of the linear combination of the eigenfunctions.

\noindent The analysis developed in this paper left some interesting open problems, and one immediate extension would be to develop general results for parabolic systems with analytic semigroup generators. There could be also interesting to consider systems of coupled PDEs. Another challenging direction involves the rapid stabilization of the multidimensional wave equation under geometric control conditions, where only observability inequalities are available.



\section*{Conflicts of Interest}
Te authors declare no conflicts of interest

\section*{Acknowledgments}

Christian Calle received financial support by ANID BECAS/DOCTORADO NACIONAL 21252433 and PIIC 048/2025  provided by the Universidad Técnica Federico Santa María. Patricio Guzmán was supported by FONDECYT 11240290. Hugo Parada is currently supported by the Agence Nationale de la Recherche through the QuBiCCS project ANR-24-CE40-3008.



\appendix

\section{Weak Spectral Inequality Proof}

\begin{proof}[Proof of the weak spectral inequality \eqref{spect_ineq}]
We argue by contradiction. Let's assume that \eqref{spect_ineq} is false, i.e,

\begin{equation}\label{cont_espec}
 \forall m \in \mathbb{N}, \exists a^m:=\left(a_1^m, a_2^m, \ldots, a_n^n\right)~:~\left\|\sum_{n=1}^N a_n^m e_n\right\|_{L^2(\omega)}^2<\frac{1}{m} \sum_{n=1}^N\left(a_n^m\right)^2 \text {. }
\end{equation}

\noindent Then, let's consider a sequence $v^m$ (renormalize)

\begin{equation*}
{v^m}:=\frac{a^m}{\left\|a^m\right\|_{\mathbb{R}^N}} \quad, \quad m=1,2, \ldots    
\end{equation*}

\noindent Note that

\begin{equation}\label{renormalize}
\left\|v^m\right\|_{\mathbb{R}^N}=\frac{\left\|a^m\right\|_{\mathbb{R}^N}}{\left\|a^m\right\|_{\mathbb{R}^N}}=1
\end{equation}

\noindent and

\begin{equation}\label{ineq_1}
\left\|\sum_{n=1}^N v_n^m e_n\right\|_{L^2(\omega)}=\left\|\sum_{n=1}^N \frac{a_n^m e_n}{\left\|a^m\right\|_{\mathbb{R}^N}}\right\|_{L^2(\omega)}^2=\frac{1}{\left\|a^m\right\|_{\mathbb{R}^N}^2}\left\|\sum_{n=1}^N a_n^m e_n\right\|<\frac{1}{m\left\|a^m\right\|_{\mathbb{R}^N}^2} \left\|a^m\right\|^2=\frac{1}{m}    . 
\end{equation}

\noindent In particular, the functions $\left\{v^m\right\}_{m \in \mathbb{N}}$ are bounded in $L^2(\omega)$. Then there exists a convergent subsequence $\left\{v^{m_j}\right\}_{j \in \mathbb{N}} \subset\left\{v^m\right\}_{m\in \mathbb{N} }$ such that

\begin{equation*}
 v^{m_j} \rightarrow v \quad \text { as } \quad m \rightarrow \infty \quad \text { in } \mathbb{R}^N,   
\end{equation*}

\noindent by using \eqref{renormalize}, we have that

\begin{equation}\label{eq1_esp_ineq}
 \lim _{m \rightarrow \infty}\left\|v^m\right\|_{\mathbb{R}^N}=\|v\|_{\mathbb{R}^N}=1   .
\end{equation}

\noindent In a similar way, taking the limit as $m \rightarrow \infty$ in \eqref{ineq_1}, we get

\begin{equation*}
\sum_{n=1}^N v_n e_n=0,\quad\text{in }\omega.    
\end{equation*}

\noindent Since $e_n$ is a continuous function, we have that there exists an uncountable set $\mathcal{M} \subset \omega$  such that

\begin{equation*}
\forall x \in \mathcal{M}, \exists n=n_x: e_n(x) \neq 0 .    
\end{equation*}

\noindent Note that if we have that $e_n=0$ in $\omega$, by the unique continuation property, we have that $e_n=0$ in $\Omega$, which is not possible.

\noindent Therefore,

\begin{equation}\label{eq2_esp_ineq}
\sum_{n=1}^N v_n e_n\left(x_0\right)=0 , \quad \forall x_0 \in \mathcal{M}.    
\end{equation}

\noindent \noindent As  $e_n\left(x_0\right) \neq 0$ for all $n=n_x$, and $v_n$ solves \eqref{eq2_esp_ineq}, we have that

\begin{equation*}
v_1, \ldots, \ldots, v_N=0,
\end{equation*}

\noindent which is a contradiction because by \eqref{eq1_esp_ineq}, we have that $\|v_n\|_{\mathbb{R}^{N}}=1$
Thus, we get

\begin{equation*}
Y_N^{\top} J_N Y_N \geq C_\lambda\left\|Y_N\right\|_2^2.    
\end{equation*}

\end{proof}


\section{Inner product}\label{App:inner_product}

\noindent We define the inner product $ (\cdot, \cdot)_{\mu} $ on $ L^2(\Omega) $ by

\begin{equation*}
(u, v)_{\mu} = \mu (P_N u, P_N v) + (P_N^\perp u, P_N^\perp v),
\end{equation*}

\noindent This bilinear form induces a norm $ \|u\|_\mu = \sqrt{(u, u)_\mu} = \sqrt{ \mu \|P_N u\|^2 + \|P_N^\perp u\|^2 } $, which is equivalent to the original norm on $ H $. Since $\mu>1$, we have $\|u\|_{L^2(\Omega)}^2 \leq(u, u)_\mu \leq \mu\|u\|_{L^2(\Omega)}^2$, it follows that $(L^2(\Omega), \langle\cdot,\cdot\rangle)_{L^2(\Omega)}$ and $(L^2(\Omega),\langle\cdot,\cdot\rangle)_\mu$ define the same topology. This inner product will be central to the preconditioning strategy developed.

\vspace{0.5cm}

\noindent We must check the axioms. Let $ u, v, w \in H $ and $ \alpha,\beta\in \mathbb{R} $.

\noindent\textbf{1.  Symmetry:}
    We need to check if $ (u, v)_\mu = (v, u)_\mu $.
   
    Since the original $ L^2 $ inner product $ (\cdot, \cdot) $ is symmetric, we have
    
    \begin{equation*}
    \begin{aligned}
    (u,v)_\mu&=\mu (P_N u, P_N v) + (P_N^\perp u, P_N^\perp v)\\[2ex]
    &=\mu (P_N v, P_N u) + (P_N^\perp v, P_N^\perp u) \\[2ex]
    &= ( v,u )_\mu.
    \end{aligned}
    \end{equation*}

\noindent    So,  it is symmetric. Note that from this property we can deduce that

    \begin{equation*}
    \begin{aligned}
        ( u,u)_\mu&=\mu ( P_N u,P_N u )_{L^2(\Omega)}+( P_N^\perp u, P_N^\perp u )_{L^2(\Omega)}\\
        &=\mu ( u,P_N u )_{L^2(\Omega)}+(  u, P_N^\perp u )_{L^2(\Omega)}\\
        &=(u, \mathscr{C} v)_{L^2(\Omega)}.
        \end{aligned}
    \end{equation*}

\noindent\textbf{2.  Linearity:}
    This follows directly from the linearity of the inner product and the linearity of the projections $ P_N $ and $ P_N^\perp $.
    
    \begin{equation*}
    \begin{aligned}
    (\alpha u + \beta v, w)_\mu &= \mu(\alpha u + \beta v, P_N w) + (\alpha u + \beta v, P_N^\perp w)  \\[2ex]
    &= \alpha[\mu (u, P_N w) + (u, P_N^\perp w)] + \beta[\mu (v, P_N w) + (v, P_N^\perp w)]  \\[2ex]
    &= \alpha (u, w)_\mu + \beta (v, w)_\mu
    \end{aligned}
    \end{equation*}
    So, it is linear.

\noindent\textbf{3.  Positive-Definiteness:}

    We need $ (u, u)_\mu \geq 0 $ for all $ u \neq 0 $.
    \begin{equation*}
    \begin{aligned}
    (u, u)_\mu &= \mu (P_N u, P_N u) + (P_N^\perp u, P_N^\perp u)= \mu \|P_N u\|^2 + \|P_N^\perp u\|^2
    \end{aligned}
    \end{equation*}
    Since $ \mu > 0 $ and norms are non-negative, $ (u, u)_\mu \ge 0 $. Furthermore, $ (u, u)_\mu = 0 $ implies $ \|P_N^\perp u\| = 0 $ and $ \|P_N u\| = 0 $. This means $ P_N u = 0 $ and $ P_N^\perp u = 0 $, which together imply $ u = 0 $.

\section[Isomorphism of C]{Isomorphism of $\mathscr{C}$}\label{iso_C}

\begin{prop}
 Let $P_N$ be the $L^2$-orthogonal projection onto the first $N$ eigenfunctions of $-\Delta$ with Dirichlet conditions. For $\mu_\lambda>0$, define $\mathscr{C}=\mu_\lambda P_N+P_N^{\perp}$. Then $\mathscr{C}: H_0^1(\Omega) \rightarrow H_0^1(\Omega)$ is a bounded linear isomorphism. 
\end{prop}

\begin{proof}
In order to prove that $\mathscr{C}$ is a bounded linear isomorphism, we have to verify it is a bounded linear bijection with bounded inverse between Hilbert spaces.

It's clear linear, let's show that it is bounded. Let $y\in H_0^1(\Omega)$, arbitrary. Then

\begin{equation*}
    \begin{aligned}
        \|\mathscr{C} y\|_{H_0^1(\Omega)}^2&=\mu\sum_{i=1}^{N} \lambda_i  y_i^2+\sum_{i=N+1}^{\infty} \lambda_i y_i^2+\mu\|P_N y\|^2_{L^2(\Omega)}+\|P_N^\perp y\|^2_{L^2(\Omega)}\\[2ex]
        &\leq \max\{\mu_\lambda,1\}\|y\|_{H_0^1(\Omega)}.
    \end{aligned}
\end{equation*}

Let's define $C^{-1}: H_0^1(\Omega)\rightarrow H_0^1(\Omega)$ given by $C^{-1} y :=\frac{1}{\mu_\lambda}P_N y + P_N^\perp y$. Then 

\begin{equation*}
    \begin{aligned}
        \|\mathscr{C}^{-1}y \|_{H_0^1(\Omega)}^2&=\frac{1}{\mu_\lambda}\sum_{i=1}^{N} \lambda_i y_i^2+\sum_{i=N+1}^{\infty} \lambda_i y_i^2+\frac{1}{\mu_\lambda}\|P_N y\|_{L^2(\Omega)}^2+\|P_N^\perp y\|_{L^2(\Omega)}^2\\[2ex]
        &\leq \max\left\lbrace\frac{1}{\mu_\lambda},1\right\rbrace\|y\|_{H_0^1(\Omega)}^2.
    \end{aligned}
\end{equation*}

\noindent We also have that $\mathscr{C C}^{-1}=I$. Indeed, 

\begin{equation*}
    \begin{aligned}
        \mathscr{CC}^{-1} y &=\mathscr{C}\left( \frac{1}{\mu_\lambda} P_N y + P_N^\perp y\right)\\[2ex]
        &= P_N y +P_N^\perp y =y .
    \end{aligned}
\end{equation*}

\noindent Since both $\mathscr{C}$ and $\mathscr{C}^{-1}$ are bounded, $\mathscr{C}$ is an isomorphism of $H_0^1$. So indeed

\begin{equation*}
 w=\mathscr{C} z \in H_0^1 \quad \text { if and only if } \quad z \in H_0^1 .   
\end{equation*}
\end{proof}



\bibliographystyle{plain}
\bibliography{references}

\begin{thebibliography}{10}

\bibitem{a2023}
G.~Arias.
\newblock Stabilization of a microbeam model with distributed disturbance.
\newblock {\em System and Control Letters}, 173:105466, 2023.

\bibitem{badra2014}
M.~Badra and T.~Takahashi.
\newblock On the fattorini criterion for approximate controllability and
  stabilizability of parabolic systems.
\newblock {\em ESAIM: Control, Optimisation and Calculus of Variations},
  20(3):924--956, 2014.

\bibitem{bmp2025}
I.~Balogoun, S.~Marx, and F.~Plestan.
\newblock Sliding {M}ode {C}ontrol for a {C}lass of {L}inear
  {I}nfinite-{D}imensional {S}ystems.
\newblock {\em IEEE Transactions on Automatic Control}, 70(5), 2025.

\bibitem{balogoun2025}
Isma{\"\i}la Balogoun, Swann Marx, and Franck Plestan.
\newblock Sliding mode control for a class of linear infinite-dimensional
  systems.
\newblock {\em IEEE Transactions on Automatic Control}, 2025.

\bibitem{b2010}
V.~Barbu.
\newblock {\em Nonlinear Differential Equations of Monotone Types in Banach
  Spaces}.
\newblock Springer Monographs in Mathematics. Springer, 2010.

\bibitem{bt2004}
V.~Barbu and R.~Triggiani.
\newblock Internal stabilization of {N}avier-{S}tokes equations with
  finite-dimensional controllers.
\newblock {\em Indiana University Mathematics Journal}, 53(5):1443--1494, 2004.

\bibitem{barbuwang2003}
V.~Barbu and G.~Wang.
\newblock Internal stabilization of semilinear parabolic systems.
\newblock {\em Journal of Mathematical Analysis and Applications},
  285(2):387--407, 2003.

\bibitem{breiten2014}
T.~Breiten and K.~Kunisch.
\newblock Riccati-based feedback control of the monodomain equations with the
  fitzhugh--nagumo model.
\newblock {\em SIAM Journal on Control and Optimization}, 52(6):4057--4081,
  2014.

\bibitem{HB2010}
H.~Brezis.
\newblock {\em Functional {A}nalysis, {S}obolev {S}paces and {P}artial
  {D}ifferential {E}quations}.
\newblock Universitext. Springer, 2010.

\bibitem{ccg2023}
R.~Capistrano-Filho, E.~Cerpa, and F.~Gallego.
\newblock Rapid exponential stabilization of a {B}oussinesq system of
  {K}d{V}-{K}d{V} type.
\newblock {\em Communications in Contemporary Mathematics}, 25(3):2150111,
  2023.

\bibitem{cc2013}
E.~Cerpa and J.-M. Coron.
\newblock Rapid stabilization for a {K}orteweg-de {V}ries equation from the
  left {D}irichlet boundary condition.
\newblock {\em IEEE Transactions on Automatic Control}, 58(7):1688--1695, 2013.

\bibitem{C2007}
J.-M. Coron.
\newblock {\em Control and nonlinearity}.
\newblock Vol. 136 of Mathematical Surveys and Monographs. American
  Mathematical Society, Providence, 2007.

\bibitem{cl2014}
J.-M. Coron and Q.~L\"u.
\newblock Local rapid stabilization for a {K}orteweg-de {V}ries equation with
  {N}eumann boundary control on the right.
\newblock {\em Journal de Math\'emathiques Pures et Appliqu\'ees},
  102(6):1080--1120, 2014.

\bibitem{cl2015}
J.-M. Coron and Q.~L\"u.
\newblock Fredholm transform and local rapid stabilization for a
  {K}uramoto-{S}ivashinsky equation.
\newblock {\em Journal of Differential Equations}, 259(8):3683--3729, 2015.

\bibitem{hcp2020}
H.~Feng, C-Z. Xu, and P-F. Yao.
\newblock Observers and disturbance rejection control for a heat equation.
\newblock {\em IEEE Transactions on Automatic Control}, 65(11):4957--4964,
  2020.

\bibitem{fx2015}
Q.-H. Fu and G.-Q. Xu.
\newblock Exponential stabilization of 1-d wave equation with distributed
  disturbance.
\newblock {\em WSEAS Transactions on Mathematics}, 14:192--201, 2015.

\bibitem{glm2021}
L.~Gagnon, P.~Lissy, and S.~Marx.
\newblock A fredholm transformation for the rapid stabilization of a degenerate
  parabolic equation.
\newblock {\em SIAM Journal on Control and Optimization}, 59(5):3828--3859,
  2021.

\bibitem{gz2024}
H.~Ge and Z.~Zhang.
\newblock Stability of {Wave} {Equation} with {Variable} {Coefficients} by
  {Boundary} {Fractional} {Dissipation} {Law}.
\newblock {\em Results in Mathematics}, 79(2), 2024.

\bibitem{gj2013}
B.-Z. Guo and F.-F. Jin.
\newblock Sliding mode and active disturbance rejection control to
  stabilization of one-dimensional anti-stable wave equations subject to
  disturbance in boundary input.
\newblock {\em IEEE Transactions on Automatic Control}, 58(5):1269--1274, 2013.

\bibitem{gk2012}
B.-Z. Guo and W.~Kang.
\newblock The {L}yapunov approach to boundary stabilization of an anti-stable
  one-dimensional wave equation with boundary disturbance.
\newblock {\em International Journal of Robust and Nonlinear Control},
  24(1):54--69, 2012.

\bibitem{gk2014}
B.-Z. Guo and W.~Kang.
\newblock Lyapunov approach to the boundary stabilisation of a beam equation
  with boundary disturbance.
\newblock {\em International Journal of Control}, 87(5):925--939, 2014.

\bibitem{gl2014}
B.-Z. Guo and J.-J. Liu.
\newblock Sliding mode control and active disturbance rejection control to the
  stabilization of one-dimensional {S}chr\"odinger equation subject to boundary
  control matched disturbance.
\newblock {\em International Journal of Robust and Nonlinear Control},
  24(16):2194--2212, 2014.

\bibitem{gh2023}
P.~Guzm{\'a}n and E.~Hern{\'a}ndez.
\newblock Stabilization of the heat equation with disturbance at the flux
  boundary condition.
\newblock {\em Mathematical Methods in the Applied Sciences},
  46(17):18035--18043, 2023.

\bibitem{gh2025}
P.~Guzm{\'a}n and E.~Hern{\'a}ndez.
\newblock Rapid stabilization of an unstable heat equation with disturbance at
  the flux boundary condition.
\newblock {\em Systems \& Control Letters}, 196, 2025.

\bibitem{gp2020}
P.~Guzm{\'a}n and C.~Prieur.
\newblock Rapid stabilization of a reaction-diffusion equation with distributed
  disturbance.
\newblock In {\em 59th IEEE Conference on Decision and Control}, Jeju Island,
  Korea, 2020.

\bibitem{g2018}
P.~Guzmán.
\newblock Energy decay of a microbeam model with a locally distributed
  nonlinear feedback control.
\newblock {\em Journal of Mathematical Analysis and Applications},
  467(1):238--252, 2018.

\bibitem{h1989}
A.~Haraux.
\newblock Une remarque sur la stabilisation de certains systèmes du deuxième
  ordre en temps.
\newblock {\em Portugaliae mathematica}, 46(3):245--258, 1989.

\bibitem{ks2008}
M.~Krstic and A.~Smyshlyaev.
\newblock {\em Boundary {C}ontrol of {PDE}s. A Course on Backstepping Designs},
  volume~16 of {\em Advances in Design and Control}.
\newblock Society for Industrial and Applied Mathematics, 2008.

\bibitem{lr2025}
M.~Labbadi and C.~Roman.
\newblock Finite-time stabilization of evolution equations with maximal
  monotone maps in {H}ilbert space.
\newblock {\em preprint}, 2025.

\bibitem{lt1992}
I.~Lasiecka and R.~Triggiani.
\newblock Uniform stabilization of the wave equation with dirichlet or neumann
  feedback control without geometric conditions.
\newblock {\em Applied Mathematics \& Optimization Applied Mathematics \&
  Optimization}, 46(3), 1992.

\bibitem{lr1995}
G.~Lebeau and L.~Robbiano.
\newblock Contr{\^o}le exact de l{\'e}quation de la chaleur.
\newblock {\em Communications in Partial Differential Equations},
  20(1-2):335--356, 1995.

\bibitem{lin1991}
F.-H. Lin.
\newblock Nodal sets of solutions of elliptic and parabolic equations.
\newblock {\em Communications on Pure and Applied Mathematics},
  44(3):287–308, 1991.

\bibitem{lions1971}
J.-L. Lions.
\newblock {\em Optimal Control of Systems Governed by Partial Differential
  Equations}.
\newblock Grundlehren der Mathematischen Wissenschaften, Vol. 170.
  Springer-Verlag, New York-Berlin, 1971.

\bibitem{lions1988}
J.-L. Lions.
\newblock {\em Contr\^olabilit\'e exacte, perturbations et stabilisation de
  syst\`emes distribu\'es. Tome~2}, volume~9 of {\em Recherches en
  Math\'ematiques Appliqu\'ees}.
\newblock Masson, Paris, 1988.

\bibitem{l2003}
W.~Liu.
\newblock Boundary feedback stabilization of an unstable heat equation.
\newblock {\em SIAM Journal on Control and Optimization}, 42(3):1033--1043,
  2003.

\bibitem{m2022}
Z.-D. Mei.
\newblock Output feedback exponential stabilization for a 1-d wave pde with
  dynamic boundary.
\newblock {\em Journal of Mathematical Analysis and Applications}, 508(1),
  2022.

\bibitem{mgk2019}
S.~Misra, G.~Gorain, and S.~Kar.
\newblock Stability of wave equation with a tip mass under unkown boundary
  external disturbance.
\newblock {\em Applied Mathematics E-Notes}, 19:128--140, 2019.

\bibitem{opu2011}
Y.~Orlov, A.~Pisano, and E.~Usai.
\newblock Exponential stabilization of the uncertain wave equation via
  distributed dynamic input extension.
\newblock {\em IEEE Transactions on Automatic Control}, 56(1):212--217, 2011.

\bibitem{ot2025}
Axel Osses and Faouzi Triki.
\newblock An improved spectral inequality for sums of eigenfunctions.
\newblock {\em Proceedings of the American Mathematical Society},
  153(03):1179--1189, 2025.

\bibitem{rbcps2018}
C.~Roman, D.~Bresch-Pietri, E.~Cerpa, C.~Prieur, and O.~Sename.
\newblock Backstepping control of a wave {PDE} with unstable source terms and
  dynamic boundary.
\newblock {\em IEEE Control Systems Letters}, 2(3):459--464, 2018.

\bibitem{X2024}
{S. Xiang}.
\newblock Quantitative rapid and finite time stabilization of the heat
  equation.
\newblock {\em ESAIM: COCV}, 30, 2024.

\bibitem{s1997}
R.~Showalter.
\newblock {\em Monotone {O}perators in {B}anach {S}pace and {N}onlinear
  {P}artial {D}ifferential {E}quations}, volume~49 of {\em Mathematical Surveys
  and Monographs}.
\newblock American Mathematical Society, 1997.

\bibitem{sck2010}
A.~Smyshlyaev, E.~Cerpa, and M.~Krstic.
\newblock Boundary stabilization of a 1-{D} wave equation with in-domain
  antidamping.
\newblock {\em SIAM Journal of Control and Optimization}, 48(6):4014--4031,
  2010.

\bibitem{sgk2009}
A.~Smyshlyaev, B.-Z. Guo, and M.~Krstic.
\newblock Arbitrary decay rate for {E}uler-{B}ernoulli beam by backstepping
  boundary feedback.
\newblock {\em IEEE Transactions on Automatic Control}, 54(5):1134--1141, 2009.

\bibitem{triggiani1975}
R.~Triggiani.
\newblock On the stabilizability problem in banach space.
\newblock {\em Journal of Mathematical Analysis and Applications},
  52(3):383--403, 1975.

\bibitem{triggiani1980}
R.~Triggiani.
\newblock Boundary feedback stabilizability of parabolic equations.
\newblock {\em Applied Mathematics and Optimization}, 6(1):201--220, 1980.

\bibitem{t1998}
L.R.T Tébou.
\newblock Stabilization of the wave equation with localized nonlinear damping.
\newblock {\em Journal of Differential Equations}, 145(2):502--524, 1998.

\bibitem{vfp2023}
{Vanspranghe, N.}, {Ferrante, F.}, and {Prieur, C.}
\newblock Stabilization of the wave equation through nonlinear dirichlet
  actuation.
\newblock {\em ESAIM: COCV}, 29, 2023.

\bibitem{vabk2024}
R.~Vazquez, J.~Auriol, F.~Bribiesca-Argomedo, and M.~Krstic.
\newblock Backstepping for partial differential equations.
\newblock {\em Preprint, arXiv:2410.15146v1}, 2024.

\bibitem{weyl}
H.~Weyl.
\newblock Das asymptotische {Verteilungsgesetz} der {Eigenwerte} linearer
  partieller {Differentialgleichungen} (mit einer {Anwendung} auf die {Theorie}
  der {Hohlraumstrahlung}).
\newblock {\em Mathematische Annalen}, 71(4):441--479, 1912.

\bibitem{zgp2025}
R-X. Zhao, B-Z. Guo, and L.~Paunonen.
\newblock Robust output regulation for multi-dimensional heat equation under
  boundary control.
\newblock {\em Automatica}, 171, 2025.

\bibitem{zj2017}
G.~Zheng and J.~Li.
\newblock Stabilization for the multi-dimensional heat equation with
  disturbance on the controller.
\newblock {\em Automatica}, 82, 2017.

\bibitem{zw2018}
H.-C. Zhou and G.~Weiss.
\newblock Output feedback exponential stabilization for one-dimensional
  unstable wave equations with boundary control matched disturbance.
\newblock {\em SIAM Journal on Control and Optimization}, 56(6):4098--4129,
  2018.

\end{thebibliography}


\end{document}